\documentclass[runningheads,envcountsect,envcountsame]{llncs}

\newif\ifignore 
\ignoretrue 

\usepackage{bm}
\usepackage{xspace}
\usepackage{url}
\usepackage{xcolor}
\usepackage{hyperref}
\hypersetup{
    colorlinks,
    linkcolor={black},
    citecolor={blue!50!black},
    urlcolor={blue!80!black}
}
\usepackage{tikzit}
\usetikzlibrary{circuits.ee.IEC}

\tikzstyle{white dot}=[inner sep=0mm, minimum size=1.5mm, draw=black, shape=circle, text depth=-0.2mm, draw=black, fill=white, tikzit category=nodes]
\tikzstyle{black dot}=[inner sep=0mm, minimum size=1.5mm, draw=black, shape=circle, draw=black, fill=black, tikzit category=nodes]
\tikzstyle{observed}=[inner sep=0mm, minimum size=5mm, draw=black, shape=circle, text depth=-0.2mm, draw=white, tikzit draw=gray, fill=white, tikzit category=dag]
\tikzstyle{latent}=[inner sep=0mm, minimum size=5mm, draw=black, shape=circle, text depth=-0.2mm, draw=black, fill=white, tikzit category=dag]
\tikzstyle{small box}=[shape=rectangle, text height=1.5ex, text depth=0.25ex, yshift=0.5mm, fill=white, draw=black, minimum height=6mm, yshift=-0.5mm, minimum width=6mm, font={\small}, tikzit category=boxes]
\tikzstyle{medium box}=[shape=rectangle, draw=black, fill=white, small box, minimum width=8mm, tikzit category=boxes]
\tikzstyle{semilarge box}=[shape=rectangle, draw=black, fill=white, small box, minimum width=12.5mm, tikzit category=boxes]
\tikzstyle{large box}=[shape=rectangle, draw=black, fill=white, small box, minimum width=15mm, tikzit category=boxes]
\tikzstyle{upground}=[circuit ee IEC, thick, ground, rotate=90, scale=1.5, inner sep=-2mm, tikzit shape=circle, tikzit fill=blue, tikzit category=points]
\tikzstyle{downground}=[circuit ee IEC, thick, ground, rotate=-90, scale=1.5, inner sep=-2mm, tikzit shape=circle, tikzit fill=green, tikzit category=points]
\tikzstyle{point}=[regular polygon, regular polygon sides=3, draw, scale=0.75, inner sep=-0.5pt, minimum width=9mm, fill=white, regular polygon rotate=180, tikzit category=points]
\tikzstyle{copoint}=[regular polygon, regular polygon sides=3, draw, scale=0.75, inner sep=-0.5pt, minimum width=9mm, fill=white, tikzit category=points]
\tikzstyle{uniform}=[point, fill=gray, tikzit shape=circle, scale=0.5]
\tikzstyle{label}=[font={\footnotesize}, text height=1.5ex, text depth=0.25ex, tikzit draw=blue, tikzit fill=white, tikzit category=labels]
\tikzstyle{left label}=[label, anchor=east, xshift=2mm, tikzit draw=green, tikzit fill=white, tikzit category=labels]
\tikzstyle{right label}=[label, anchor=west, xshift=-2mm, tikzit draw=purple, tikzit fill=white, tikzit category=labels]
\tikzstyle{disintegration}=[draw=black, fill={gray!50}, tikzit fill=gray, shape=rectangle, minimum width=1.6cm, minimum height=1.2cm, opacity=0.3]
\tikzstyle{empty diag}=[shape=rectangle, draw=darkgray, dashed, minimum width=8mm, minimum height=8mm, yshift=0.5mm]

\tikzstyle{diredge}=[->, >=latex]
\tikzstyle{dashed edge}=[-, dashed]

\usepackage{amsmath,amssymb,mathtools}
\usepackage{wrapfig}
\spnewtheorem*{theorem*}{Theorem}{\bfseries}{\itshape}
\spnewtheorem*{proposition*}{Proposition}{\bfseries}{\itshape}

\newcommand{\counit}{%
\,\begin{tikzpicture}[yshift=-1mm]
\node [black dot] (a) at (0,0.35) {}; 
\draw (0,-0.3)--(a);
\end{tikzpicture}\,\xspace}

\newcommand{\uniform}{%
\,\begin{tikzpicture}[yshift=1.5mm]
\node [uniform] (a) at (0,-0.35) {}; 
\draw (a)--(0,0.3);
\end{tikzpicture}\,\xspace}
\newcommand{\comult}{%
\,\begin{tikzpicture}[yshift=0.5mm]
\node [black dot] (a) {};
\draw (-90:0.55)--(a);
\draw (a) -- (45:0.6);
\draw (a) -- (135:0.6);
\end{tikzpicture}\,\xspace}
\newcommand{\comultdots}{%
\,\begin{tikzpicture}[yshift=0.5mm]
\node [black dot] (a) at (0,0) {};
\node (b) at (0,0.4) {\footnotesize ...};
\draw (-90:0.55)--(a);
\draw (a) -- (45:0.6);
\draw (a) -- (135:0.6);
\end{tikzpicture}\,\xspace}

\newcommand{\QEDbox}{\square}
\newcommand{\QED}{\hspace*{\fill}$\QEDbox$}

\newcommand{\mat}{\bm}
\newcommand{\freeCDU}[1]{\ensuremath{\mathsf{FreeCDU}}(#1)}
\newcommand{\free}[1]{\ensuremath{\mathsf{Syn}_{\scriptscriptstyle #1}}\xspace}
\newcommand{\tns}{\ensuremath{\otimes}}
\newcommand{\catC}{\ensuremath{\mathsf{C}}}
\newcommand{\funF}{\ensuremath{\mathcal{F}}}
\newcommand{\Stoch}{\ensuremath{\mathsf{Stoch}}\xspace}
\newcommand{\BN}[1]{\ensuremath{\mathsf{BN}_{\scriptscriptstyle{#1}}\xspace}}
\newcommand{\MatRp}{\ensuremath{\mathsf{Mat}(\mathbb R^{\scriptscriptstyle +})}\xspace}

\newcommand{\cut}[1]{\ensuremath{\mathsf{Cut}_{\scriptscriptstyle #1}}}
\newcommand{\syncut}[1]{\ensuremath{\mathsf{cut}_{\scriptscriptstyle #1}}}

\newcommand{\syn}[1]{\overline{#1}}

\begin{document}

\title{%
    Causal Inference by String Diagram Surgery
}

\author{%
    Bart Jacobs\inst{1}
    \and
   Aleks Kissinger\inst{1}
    \and
    Fabio Zanasi\inst{2}
}

\authorrunning{B. Jacobs, A. Kissinger and F. Zanasi}

\institute{
    Radboud University, Nijmegen, The Netherlands   
     \and
    University College London, London, United Kingdom
}

\maketitle

\begin{abstract}
Extracting causal relationships from observed correlations is a growing area
in probabilistic reasoning, originating with the seminal work of Pearl and
others from the early 1990s. This paper develops a new, categorically oriented
view based on a clear distinction between syntax (string diagrams) and
semantics (stochastic matrices), connected via interpretations as
structure-preserving functors.

A key notion in the identification of causal effects is that of an
intervention, whereby a variable is forcefully set to a particular value
independent of any prior dependencies. We represent the effect of such an
intervention as an endofunctor which performs `string diagram surgery' within
the syntactic category of string diagrams. This diagram surgery in turn yields
a new, interventional distribution via the interpretation functor. While in
general there is no way to compute interventional distributions purely from
observed data, we show that this is possible in certain special cases using a
calculational tool called comb disintegration.

We demonstrate the use of this technique on a well-known toy example, where we
predict the causal effect of smoking on cancer in the presence of a
confounding common cause. After developing this specific example, we show this
technique provides simple sufficient conditions for computing interventions
which apply to a wide variety of situations considered in the causal inference
literature.

\end{abstract}

\section{Introduction}\label{sec:intro}

An important conceptual tool for distinguishing correlation from
causation is the possibility of \textit{intervention}. For example, a
randomised drug trial attempts to destroy any confounding `common
cause' explanation for correlations between drug use and recovery by
randomly assigning a patient to the control or treatment group,
independent of any background factors. In an ideal setting, the
observed correlations of such a trial will reflect genuine causal
influence.
Unfortunately, it is not always possible (or ethical) to ascertain
causal effects by means of actual interventions. For instance, one is
unlikely to get approval to run a clinical trial on whether smoking
causes cancer by randomly assigning 50\% of the patients to smoke, and
waiting a bit to see who gets cancer. However, in certain situations
it is possible to predict the effect of such a hypothetical
intervention from purely observational data.

In this paper, we will focus on the problem of \textit{causal
identifiability}. For this problem, we are given observational data as
a joint distribution on a set of variables and we are furthermore
provided with a \textit{causal structure} associated with those
variables. This structure, which typically takes the form of a
directed acyclic graph or some variation thereof, tells us which
variables can in principle have a causal influence on others. The
problem then becomes whether we can measure how strong those causal
influences are, by means of computing an \textit{interventional}
distribution. That is, can we ascertain what would have happened if a
(hypothetical) intervention had occurred?

Over the past 3 decades, a great deal of work has been done in
identifying necessary and sufficient conditions for causal
identifiability in various special cases, starting with very specific
notions such as the \textit{back-door} and \textit{front-door}
criteria~\cite{PearlBook} and progressing to more general necessary and sufficient
conditions for causal identifiability based on the
\textbf{do}-calculus~\cite{DoComplete}, or combinatoric concepts such as
confounded components in semi-Makovian models~\cite{TianPearl,ShpitserPearl}.

This style of causal reasoning relies crucially on a delicate interplay
between syntax and semantics, which is often not made explicit in the literature.
The syntactic object of interest is the causal structure (e.g. a causal
graph), which captures something about our understanding of the world, and the
mechanisms which gave rise to some observed phenomena. The semantic object of
interest is the data: joint and conditional probability distributions on some
variables. Fixing a causal structure entails certain constraints on which
probability distributions can arise, hence it is natural to see distributions
satisfying those constraints as models of the syntax.

In this paper, we make this interplay precise using functorial semantics in
the spirit of Lawvere~\cite{lawvere1963functorial}, and develop basic
syntactic and semantic tools for causal reasoning in this setting. We take as
our starting point a functorial presentation of Bayesian networks similar to
the one appearing in~\cite{Fong12}, which in turn is based on the diagrammatic
presentation given in~\cite{CSBayes}. The syntactic role is played by string
diagrams, which give an intuitive way to represent morphisms of a monoidal
category as boxes plugged together by wires. Given a directed acyclic graph
(dag) $G$, we can form a free category $\free{G}$ whose arrows are (formal)
string diagrams which represent the causal structure syntactically. Structure-
preserving functors from $\free{G}$ to \Stoch, the category of stochastic
matrices, then correspond exactly to Bayesian networks based on the dag
$G$.

Within this framework, we develop the notion of intervention as an
operation of `string diagram surgery'. Intuitively, this cuts a string
diagram at a certain variable, severing its link to the
past. Formally, this is represented as an endofunctor on the syntactic
category $\syncut{X} \colon \free{G} \to \free{G}$, which propagates
through a model $\mathcal F \colon \free{G} \to \Stoch$ to send
observational probabilities $\mathcal F(\omega)$ to interventional
probabilities $\mathcal F(\syncut{X}(\omega))$.

The $\syncut{X}$ endofunctor gives us a diagrammatic means of
computing interventional distributions given complete knowledge of
$\mathcal F$.  However, more interestingly, we can sometimes compute
interventionals given only partial knowledge of $\mathcal F$, namely
some observational data. We show that this can also be done via a
technique we call \textit{comb disintegration}, which is a string
diagrammatic version of a technique called \textit{c-factorisation}
introduced by Tian and Pearl~\cite{TianPearl}. Our approach
generalises disintegration, a calculational tool whereby a joint state
on two variables is factored into a single-variable state and a
channel, representing the marginal and conditional parts of the
distribution, respectively. Disintegration has recently been
formulated categorically in~\cite{ClercDDG17} and using string
diagrams in~\cite{ChoJ17a}. We take the latter as a starting point,
but instead consider a factorisation of a three-variable state into a
channel and a \textit{comb}. The latter is a special kind of map which
allows inputs and outputs to be interleaved. They were originally
studied in the context of quantum communication protocols, seen as
games~\cite{QGames}, but have recently been used extensively in the
study of causally-ordered quantum~\cite{comb,PauloHierarchy} and
generalised~\cite{KissingerUijlen} processes. While originally
imagined for quantum processes, the categorical formulation given
in~\cite{KissingerUijlen} makes sense in both the classical case
(\Stoch) and the quantum. Much like Tian and Pearl's technique, comb
factorisation allows one to characterise when the confounding parts of
a causal structure are suitably isolated from each other, then exploit
that isolation to perform the concrete calculation of interventional
distributions.

However, unlike in the traditional formulation, the syntactic and
semantic aspects of causal identifiability within our framework
exactly mirror one-another.  Namely, we can give conditions for causal
identifiability in terms of factorisation a morphism in \free{G},
whereas the actual concrete computation of the interventional
distribution involves factorisation of its interpretation in
\Stoch. Thanks to the functorial semantics, the former immediately
implies the latter.

To introduce the framework, we make use of a running example taken from
Pearl's book~\cite{PearlBook}: identifying the causal effect of smoking on
cancer with the help of an auxiliary variable (the presence of tar in the
lungs). After providing some preliminaries on stochastic matrices and the
functorial presentation of Bayesian networks in Sections \ref{sec:stoch} and
\ref{sec:bn}, we introduce the smoking example in Section~\ref{sec:smoking}.
In Section~\ref{sec:surgery} we formalise the notion of intervention as string
diagram surgery, and in Section~\ref{sec:comb} we introduce the combs and
prove our main calculational result: the existence and uniqueness of comb
factorisations. In Section~\ref{sec:smoking-return}, we show how to apply this
theorem in computing the interventional distribution in the smoking example,
and in \ref{sec:general}, we show how this theorem can be applied in a more
general case which captures (and slightly generalises) the conditions given
in~\cite{TianPearl}. In Section~\ref{sec:conclusion}, we conclude and describe
several avenues of future work.

\section{Stochastic Matrices and Conditional Probabilities}\label{sec:stoch}

Symmetric monoidal categories (SMCs) give a very general setting for
studying processes which can be composed in sequence (via the usual
categorical composition $\circ$) and in parallel (via the monoidal
composition $\tns$). Throughout this paper, we will use \textit{string
diagram} notation~\cite{Sel2009-graphical} for depicting composition
of morphisms in an SMC. In this notation, morphisms are depicted as
boxes with labelled input and output wires, composition $\circ$ as
`plugging' boxes together, and the monoidal product $\tns$ as placing
boxes side-by-side. Identitiy morphisms are depicted simply as a wire
and the unit $I$ of $\tns$ as the empty diagram. The `symmetric' part
of the structure consists of symmetry morphisms, which enable us to
permute inputs and outputs arbitrarily. We depict these as wire-
crossings: \begin{tikzpicture}
	\begin{pgfonlayer}{nodelayer}
		\node [style=none] (25) at (0.5, -0.5) {};
		\node [style=none] (29) at (-0.25, 0.25) {};
		\node [style=none] (40) at (0.5, -0.25) {};
		\node [style=none] (41) at (-0.25, -0.5) {};
		\node [style=none] (42) at (0.5, 0.25) {};
		\node [style=none] (44) at (-0.25, -0.25) {};
		\node [style=none] (45) at (0.5, 0.5) {};
		\node [style=none] (46) at (-0.25, 0.5) {};
	\end{pgfonlayer}
	\begin{pgfonlayer}{edgelayer}
		\draw [in=-90, out=-90, looseness=0.75] (40.center) to (25.center);
		\draw (40.center) to (29.center);
		\draw [in=-90, out=-90, looseness=0.75] (44.center) to (41.center);
		\draw (44.center) to (42.center);
		\draw (46.center) to (29.center);
		\draw (45.center) to (42.center);
	\end{pgfonlayer}
\end{tikzpicture}. Morphisms whose domain is $I$ are
called \textit{states}, and they will play a special role throughout
this paper.



A monoidal category of prime interest in this paper is
$\Stoch$, whose objects are finite sets and morphisms $\mat f : A \to B$
are $|B| \times |A|$ dimensional stochastic matrices. That is, they
are matrices of positive numbers (including $0$) whose columns each sum to 1:
\[ \mat f = \{ \mat f_i^j \in \mathbb R^{\scriptscriptstyle +} \ |\  i \in A, j \in B \} 
\qquad\mbox{with}\qquad 
\textstyle\sum_j \mat f_i^j = 1, \mbox{ for all }i. \]

\noindent Note we adopt the physicists convention of writing row
indices as superscripts and column indices as subscripts.  Stochastic
matrices are of interest for probabilistic reasoning, because they
exactly capture the data of a conditional probability
distribution. That is, if we take $A := \{1, \ldots, m\}$ and $B :=
\{1, \ldots, n\}$, conditional probabilities naturally arrange
themselves into a stochastic matrix: 
\[
\mat f_i^j := P(B=j|A=i) \ \ \leadsto\ \ 
\mat f = 
\textrm{\footnotesize
$\left(\begin{matrix}
P(B=1|A=1) & \cdots & P(B=1|A=m) \\
\vdots & \ddots & \vdots \\
P(B=n|A=1) & \cdots & P(B=n|A=m)
\end{matrix}\right)$
}
\]

States, i.e. stochastic matrices from a trivial input $I :=\{*\}$, are
(non-conditional) probability distributions, represented as column
vectors. There is only one stochastic matrix with trivial output: the row vector consisting
only of $1$'s. The latter, with notation $\counit$ as on the right,
will play a special role in this paper (see \eqref{eq:stoch-copy-uniform} below).

Composition of stochastic matrices is matrix multiplication. In terms
of conditional probabilities, this corresponds to multiplication,
followed by marginalization over the shared variable: $\sum_B
P(C|B)P(B|A)$. Identities are therefore given by identity
matrices, which we will often express in terms of the Kronecker delta
function $\mat \delta_i^j$.

The monoidal product $\otimes$ in \Stoch is the cartesian product on
objects, and Kronecker product of matrices: $(\mat f \otimes \mat g)_{(i,j)}^{(k,l)} := \mat f_i^k \mat g_j^l$. We will
typically omit parentheses and commas in the indices, writing
e.g. $\mat h_{ij}^{kl}$ instead of $\mat h_{(i,j)}^{(k,l)}$ for an arbitrary matrix entry of $\mat h \colon A
\otimes B \to C \otimes D$. In terms of conditional probabilities,
the Kronecker product corresponds to taking product distributions. That is, if $\mat f$
represents the conditional probabilities $P(B|A)$ and $\mat g$ the
probabilities $P(D|C)$, then $\mat f\otimes \mat g$ represents $P(B|A)P(D|C)$.  \Stoch
also comes with a natural choice of `swap' matrices $\mat \sigma : A
\otimes B \to B \otimes A$ given by $\mat \sigma_{ij}^{kl} := \mat \delta_i^l
\mat \delta_j^k$, making it into a symmetric monoidal category.  Every object $A$ in \Stoch has three
other pieces of structure which will play a key role in our formulation of
Bayesian networks and interventions: the \textit{copy} map, the \textit{discarding} map, and the \textit{uniform state}:
\begin{equation}\label{eq:stoch-copy-uniform}
\left(\comult\right)_i^{jk} := \mat\delta_i^j \mat\delta_i^k
\qquad\qquad
\left(\counit\right)_i := 1
\qquad\qquad
\left(\uniform\right)^i := \frac{1}{|A|}
\end{equation}

\noindent Abstractly, this provides $\Stoch$ with the structure of a \textit{CDU
  category}.

\begin{definition}
A \emph{CDU category} (for \textbf{c}opy, \textbf{d}iscard,
\textbf{u}niform) is a symmetric monoidal category $(\catC, \tns, I)$
where each object $A$ has a copy map $\comult : A \rightarrow A \tns
A$, a discarding map $\counit : A \rightarrow I$, and a uniform state
$\uniform : I \rightarrow A$ satisfying the following equations:
\begin{equation}\label{eq:CDUaxioms}
	\begin{tikzpicture}
	\begin{pgfonlayer}{nodelayer}
		\node [style=black dot] (0) at (-2, -0.5) {};
		\node [style=none] (1) at (-2, -1.5) {};
		\node [style=black dot] (2) at (-2.75, 0.25) {};
		\node [style=none] (3) at (-3.5, 1.25) {};
		\node [style=none] (4) at (-2, 1.25) {};
		\node [style=none] (5) at (-0.75, 1.25) {};
		\node [style=none] (12) at (0, 0) {$=$};
		\node [style=black dot] (13) at (2, -0.5) {};
		\node [style=none] (14) at (2, -1.5) {};
		\node [style=black dot] (15) at (2.75, 0.25) {};
		\node [style=none] (16) at (3.5, 1.25) {};
		\node [style=none] (17) at (2, 1.25) {};
		\node [style=none] (18) at (0.75, 1.25) {};
	\end{pgfonlayer}
	\begin{pgfonlayer}{edgelayer}
		\draw (1.center) to (0);
		\draw [bend left] (0) to (2);
		\draw [bend right] (2) to (4.center);
		\draw [bend left] (2) to (3.center);
		\draw [bend right] (0) to (5.center);
		\draw (14.center) to (13);
		\draw [bend right] (13) to (15);
		\draw [bend left] (15) to (17.center);
		\draw [bend right] (15) to (16.center);
		\draw [bend left] (13) to (18.center);
	\end{pgfonlayer}
\end{tikzpicture} \qquad \begin{tikzpicture}
	\begin{pgfonlayer}{nodelayer}
		\node [style=black dot] (6) at (-1.75, -0.5) {};
		\node [style=none] (7) at (-1.75, -1.5) {};
		\node [style=black dot] (8) at (-1.25, 0.25) {};
		\node [style=none] (11) at (-2.5, 0.75) {};
		\node [style=none] (13) at (-0.25, 0) {$=$};
		\node [style=none] (14) at (0.5, 1.25) {};
		\node [style=none] (15) at (0.5, -1.5) {};
	\end{pgfonlayer}
	\begin{pgfonlayer}{edgelayer}
		\draw (7.center) to (6);
		\draw [in=-90, out=30] (6) to (8);
		\draw [in=-90, out=146] (6) to (11.center);
		\draw (15.center) to (14.center);
	\end{pgfonlayer}
\end{tikzpicture} \qquad \begin{tikzpicture}
	\begin{pgfonlayer}{nodelayer}
		\node [style=none] (1) at (-2.5, -1.5) {};
		\node [style=black dot] (2) at (-2.5, -0.5) {};
		\node [style=none] (3) at (-3, 0) {};
		\node [style=none] (4) at (-2, 0) {};
		\node [style=none] (12) at (-1, 0) {$=$};
		\node [style=none] (13) at (-2, 1) {};
		\node [style=none] (14) at (-3, 1) {};
		\node [style=none] (15) at (0.25, -1.5) {};
		\node [style=black dot] (16) at (0.25, -0.5) {};
		\node [style=none] (17) at (-0.5, 0.75) {};
		\node [style=none] (18) at (1, 0.75) {};
	\end{pgfonlayer}
	\begin{pgfonlayer}{edgelayer}
		\draw [in=-90, out=23] (2) to (4.center);
		\draw [in=-90, out=158] (2) to (3.center);
		\draw (1.center) to (2);
		\draw [in=270, out=89] (4.center) to (14.center);
		\draw [in=270, out=89] (3.center) to (13.center);
		\draw [in=-90, out=23] (16) to (18.center);
		\draw [in=-90, out=158] (16) to (17.center);
		\draw (15.center) to (16);
	\end{pgfonlayer}
\end{tikzpicture} \qquad \begin{tikzpicture}
	\begin{pgfonlayer}{nodelayer}
		\node [style=uniform] (1) at (-1, -0.5) {};
		\node [style=black dot] (2) at (-1, 0.5) {};
		\node [style=none] (12) at (0, 0) {$=$};
		\node [style=empty diag] (13) at (1.5, 0) {};
	\end{pgfonlayer}
	\begin{pgfonlayer}{edgelayer}
		\draw (1) to (2);
	\end{pgfonlayer}
\end{tikzpicture}

\end{equation}
\noindent \emph{CDU functors} are symmetric monoidal functors between
CDU categories preserving copy maps, discard maps and uniform states.
\end{definition}

We assume that the CDU structure on $I$ is trivial and the CDU
structure on $A \otimes B$ is constructed in the obvious way from the
structure on $A$ and $B$. We also use the first equation in
\eqref{eq:CDUaxioms} to justify writing `copy' maps with arbitrarily
many output wires: \comultdots.

Similar to \cite{BonchiSZ18}, we can form the free CDU category
$\freeCDU{X,\Sigma}$ over a pair $(X, \Sigma)$ of a generating set of
objects $X$ and a generating set $\Sigma$ of typed morphisms $f \colon
u \to w$, with $u,w \in X^{\star}$ as follows. The category
$\freeCDU{X,\Sigma}$ has $X^{\star}$ as set of objects, and morphisms
the string diagrams constructed from the elements of $\Sigma$ and maps
$\scalebox{0.6}{\comult} \colon x \to x \tns x$,
$\scalebox{0.6}{\counit} \colon x \to I$ and $\scalebox{0.6}{\uniform}
\colon I \to x$ for each $x \in X$, taken modulo the equations
\eqref{eq:CDUaxioms}.

\begin{lemma}
$\Stoch$ is a CDU category, with CDU structure defined as in \eqref{eq:stoch-copy-uniform}.
\end{lemma}

\begin{wrapfigure}{r}{0pt}
\begin{minipage}{9em}
\vspace{-5mm}
\begin{equation}
\label{eq:discarding-final}
\begin{tikzpicture}
	\begin{pgfonlayer}{nodelayer}
		\node [style=small box] (0) at (-1, 0) {$\mat f$};
		\node [style=black dot] (1) at (-1, 1.5) {};
		\node [style=black dot] (2) at (1.75, 0.5) {};
		\node [style=none] (3) at (0.5, 0) {$=$};
		\node [style=none] (4) at (-1, -1.25) {};
		\node [style=none] (5) at (1.75, -1.25) {};
		\node [style=right label] (6) at (-0.75, -1) {$A$};
		\node [style=right label] (7) at (-0.75, 1) {$B$};
		\node [style=right label] (8) at (2, -0.25) {$B$};
	\end{pgfonlayer}
	\begin{pgfonlayer}{edgelayer}
		\draw (4.center) to (0);
		\draw (0) to (1);
		\draw (5.center) to (2);
	\end{pgfonlayer}
\end{tikzpicture}
\end{equation}
\end{minipage}
\end{wrapfigure}
An important feature of $\Stoch$ is that $I = \{ \star \}$ is the
final object, with $\counit \colon B \to I$ the map provided by the
universal property, for any set $B$. This yields
equation~\eqref{eq:discarding-final} on the right, for any $\mat f \colon A
\to B$, justifying the name ``discarding map'' for $\counit$.

We conclude by recording another significant feature of $\Stoch$:
\emph{disintegration}~\cite{ClercDDG17,ChoJ17a}. In probability
theory, this is the mechanism of factoring a joint probability
distribution $P(AB)$ as a product of the first marginal $P(A)$ and a
conditional distribution $P(B|A)$. We recall from
\cite{ChoJ17a} the string diagrammatic rendition of this
process. We say that a morphism $\mat f \colon X \to Y$ in $\Stoch$ has
\emph{full support} if, as a stochastic matrix, it has no zero
entries. When $\mat f$ is a state, it is a standard result that full
support ensures uniqueness of disintegrations of~$\mat f$.

\begin{proposition}[Disintegration]
\label{thm:disintStoch}
For any state $\mat\omega \colon I \to A \otimes B$ with full support,
there exists unique morphisms $\mat a \colon I \to A, \mat b \colon A
\to B$ such that:
\begin{equation}\label{eq:disint-def}
\scalebox{0.8}{\begin{tikzpicture}
	\begin{pgfonlayer}{nodelayer}
		\node [style=small box] (1) at (3.5, 0.75) {$\mat b$};
		\node [style=none] (2) at (1.5, 0.75) {};
		\node [style=none] (4) at (0, 0) {$=$};
		\node [style=black dot] (7) at (2.5, -0.5) {};
		\node [style=small box] (8) at (2.5, -1.75) {$\mat a$};
		\node [style=none] (9) at (1.5, 2.25) {};
		\node [style=none] (10) at (3.5, 2.25) {};
		\node [style=medium box] (11) at (-2.25, -0.25) {$\mat\omega$};
		\node [style=none] (12) at (-2.75, 0.25) {};
		\node [style=none] (13) at (-2.75, 1.25) {};
		\node [style=none] (14) at (-1.75, 0.25) {};
		\node [style=none] (15) at (-1.75, 1.25) {};
		\node [style=right label] (16) at (-2.5, 1) {$A$};
		\node [style=right label] (17) at (1.75, 1.75) {$A$};
		\node [style=right label] (18) at (3.75, 1.75) {$B$};
		\node [style=right label] (19) at (-1.5, 1) {$B$};
	\end{pgfonlayer}
	\begin{pgfonlayer}{edgelayer}
		\draw (7) to (8);
		\draw [in=-90, out=15] (7) to (1);
		\draw [in=-90, out=165] (7) to (2.center);
		\draw (2.center) to (9.center);
		\draw (1) to (10.center);
		\draw (12.center) to (13.center);
		\draw (14.center) to (15.center);
	\end{pgfonlayer}
\end{tikzpicture}}
\end{equation}
\end{proposition}

\noindent Note that equation \eqref{eq:discarding-final} and the CDU
rules immediately imply that the unique $\mat a \colon I \to A$ in
Proposition~\ref{thm:disintStoch} is the marginal of $\mat\omega$ onto
$A$: $\scalebox{0.6}{\begin{tikzpicture}
	\begin{pgfonlayer}{nodelayer}
		\node [style=none] (12) at (0.5, 0) {};
		\node [style=none] (13) at (0.5, 1) {};
		\node [style=right label] (16) at (0.75, 0.5) {$B$};
		\node [style=none] (20) at (-0.5, 0) {};
		\node [style=none] (21) at (-0.5, 1) {};
		\node [style=right label] (22) at (-0.25, 0.5) {$A$};
		\node [style=black dot] (23) at (0.5, 1) {};
		\node [style=medium box] (41) at (0, -0.5) {$\mat\omega$};
	\end{pgfonlayer}
	\begin{pgfonlayer}{edgelayer}
		\draw (12.center) to (13.center);
		\draw (20.center) to (21.center);
	\end{pgfonlayer}
\end{tikzpicture}}$.

\section{Bayesian Networks as String Diagrams}\label{sec:bn}

Bayesian networks are a widely-used tool in probabilistic
reasoning. They give a succinct representation of conditional
(in)dependences between variables as a directed acyclic
graph. Traditionally, a Bayesian network on a set of variables $A, B,
C, \ldots$ is defined as a directed acyclic graph (dag) $G$, an
assignment of sets to each of the nodes $V_G := \{A, B, C, \ldots\}$
of $G$ and a joint probability distribution over those variables which
factorises as $P(V_G) = \prod_{A \in V_G} P(A\,|\,\textrm{Pa}(A))$ where
$\textrm{Pa}(A)$ is the set of parents of $A$ in $G$. Any joint
distribution that factorises this way is said to satisfy the \textit{global Markov property} with respect to the dag $G$.
Alternatively, a Bayesian network can be seen as a dag equipped
with a set of conditional probabilities $\{ P(A \,|\, \textrm{Pa}(A)) \mid  A \in V_G \}$
which can be combined to form the joint state. Thanks to disintegration, these
two perspectives are equivalent.

Much like in the case of disintegration in the previous section,
Bayesian networks have a neat categorical description as string
diagrams in the category $\Stoch$
\cite{CSBayes,Fong12,JacobsZ16,JacobsZ18}. For example, here is a Bayesian
network in its traditional depiction as a dag with an associated joint
distribution over its vertices, and as a string diagram
in \Stoch:
\begin{equation*} 
\begin{array}{ccc}
\begin{array}{c}
\scalebox{0.8}{\begin{tikzpicture}
	\begin{pgfonlayer}{nodelayer}
		\node [style=observed] (0) at (0, -1.75) {$A$};
		\node [style=observed] (1) at (-1.5, 0) {$B$};
		\node [style=observed] (2) at (1.5, 0) {$D$};
		\node [style=observed] (3) at (0, 1.75) {$C$};
		\node [style=observed] (4) at (3, 1.75) {$E$};
	\end{pgfonlayer}
	\begin{pgfonlayer}{edgelayer}
		\draw [style=diredge] (0) to (1);
		\draw [style=diredge] (0) to (2);
		\draw [style=diredge] (2) to (4);
		\draw [style=diredge] (2) to (3);
		\draw [style=diredge] (1) to (3);
	\end{pgfonlayer}
\end{tikzpicture}}
\\[+3em]
P(ABCDE) = 
\\
P(A)P(B|A)P(D|A)P(C|BD)P(E|D)
\end{array}
& \hspace*{2em} &
\scalebox{0.8}{\begin{tikzpicture}
	\begin{pgfonlayer}{nodelayer}
		\node [style=small box] (0) at (1.5, -3.25) {$\mat a$};
		\node [style=black dot] (1) at (1.5, -1.75) {};
		\node [style=right label] (4) at (1.75, -2.25) {$A$};
		\node [style=small box] (5) at (3, 0) {$\mat d$};
		\node [style=small box] (6) at (0.25, 0) {$\mat b$};
		\node [style=black dot] (7) at (3, 1.5) {};
		\node [style=black dot] (8) at (0.25, 1.5) {};
		\node [style=right label] (9) at (3.25, 1) {$D$};
		\node [style=right label] (10) at (0.5, 1) {$B$};
		\node [style=none] (11) at (-2.25, 0) {};
		\node [style=small box] (12) at (1.5, 2.75) {$\mat c$};
		\node [style=none] (13) at (1.25, 2.25) {};
		\node [style=none] (14) at (1.75, 2.25) {};
		\node [style=none] (15) at (-0.75, 2.25) {};
		\node [style=none] (16) at (3, 4.25) {};
		\node [style=none] (17) at (1.5, 4.25) {};
		\node [style=small box] (18) at (4.25, 2.75) {$\mat e$};
		\node [style=none] (19) at (4.25, 2.25) {};
		\node [style=none] (20) at (-0.75, 4.25) {};
		\node [style=none] (21) at (4.25, 4.25) {};
		\node [style=none] (22) at (-2.25, 4.25) {};
		\node [style=right label] (23) at (-2, 4) {$A$};
		\node [style=right label] (24) at (-0.5, 4) {$B$};
		\node [style=right label] (25) at (4.5, 4) {$E$};
		\node [style=right label] (26) at (1.75, 4) {$C$};
		\node [style=right label] (27) at (3.25, 4) {$D$};
		\node [style=none] (28) at (0.25, -0.25) {};
		\node [style=none] (29) at (3, -0.25) {};
	\end{pgfonlayer}
	\begin{pgfonlayer}{edgelayer}
		\draw (0) to (1);
		\draw (5) to (7);
		\draw (6) to (8);
		\draw [in=-90, out=-180, looseness=0.75] (1) to (11.center);
		\draw [in=-90, out=180] (7) to (14.center);
		\draw [in=-90, out=0] (7) to (19.center);
		\draw (12) to (17.center);
		\draw (7) to (16.center);
		\draw [in=-90, out=-180] (8) to (15.center);
		\draw (18) to (21.center);
		\draw (15.center) to (20.center);
		\draw (11.center) to (22.center);
		\draw [in=-90, out=0] (8) to (13.center);
		\draw [in=-90, out=165] (1) to (28.center);
		\draw [in=-90, out=15] (1) to (29.center);
	\end{pgfonlayer}
\end{tikzpicture}}
\end{array}
\end{equation*}
In the string diagram above, the stochastic matrix $\mat a \colon I \to A$ contains the probabilities $P(A)$, $\mat b \colon B \to A$ contains the conditional probabilities $P(B|A)$, $\mat c \colon B \otimes D \to C$ contains $P(C|BD)$, and so on. The entire diagram is then equal to a state $\mat \omega \colon I \to A \otimes B \otimes C \otimes D \otimes E$ which represents $P(ABCDE)$.

Note the dag and the diagram above look similar in structure. The main
difference is the use of copy maps to make each variable (even those
that are not leaves of the dag, $A$, $B$ and $D$) an output of the
overall diagram.  This corresponds to a variable being
\textit{observed}. We can also consider Bayesian networks with
\textit{latent} variables, which do not appear in the joint
distribution due to marginalisation. Continuing the example above, making
$A$ into a latent variable yields the following depiction as a string
diagram:
\[ \begin{array}{ccc}
\begin{array}{c}
\scalebox{0.8}{\begin{tikzpicture}
	\begin{pgfonlayer}{nodelayer}
		\node [style=latent] (0) at (0, -1.75) {$A$};
		\node [style=observed] (1) at (-1.5, 0) {$B$};
		\node [style=observed] (2) at (1.5, 0) {$D$};
		\node [style=observed] (3) at (0, 1.75) {$C$};
		\node [style=observed] (4) at (3, 1.75) {$E$};
	\end{pgfonlayer}
	\begin{pgfonlayer}{edgelayer}
		\draw [style=diredge] (0) to (1);
		\draw [style=diredge] (0) to (2);
		\draw [style=diredge] (2) to (4);
		\draw [style=diredge] (2) to (3);
		\draw [style=diredge] (1) to (3);
	\end{pgfonlayer}
\end{tikzpicture}}
\\[+3em]
P(BCDE) = 
\\
\sum_A P(A)P(B|A)P(D|A)P(C|BD)P(E|D)
\end{array}
& \hspace*{1.5em} &
\scalebox{0.8}{\begin{tikzpicture}
	\begin{pgfonlayer}{nodelayer}
		\node [style=small box] (0) at (0.125, -3.25) {$\mat a$};
		\node [style=black dot] (1) at (0.125, -1.75) {};
		\node [style=right label] (4) at (0.375, -2.25) {$A$};
		\node [style=small box] (5) at (1.5, 0) {$\mat d$};
		\node [style=small box] (6) at (-1.25, 0) {$\mat b$};
		\node [style=black dot] (7) at (1.5, 1.5) {};
		\node [style=black dot] (8) at (-1.25, 1.5) {};
		\node [style=right label] (9) at (1.75, 1) {$D$};
		\node [style=right label] (10) at (-1, 1) {$B$};
		\node [style=small box] (12) at (0, 2.75) {$\mat c$};
		\node [style=none] (13) at (-0.25, 2.25) {};
		\node [style=none] (14) at (0.25, 2.25) {};
		\node [style=none] (15) at (-2.25, 2.25) {};
		\node [style=none] (16) at (1.5, 4.25) {};
		\node [style=none] (17) at (0, 4.25) {};
		\node [style=small box] (18) at (2.75, 2.75) {$\mat e$};
		\node [style=none] (19) at (2.75, 2.25) {};
		\node [style=none] (20) at (-2.25, 4.25) {};
		\node [style=none] (21) at (2.75, 4.25) {};
		\node [style=right label] (24) at (-2, 4) {$B$};
		\node [style=right label] (25) at (3, 4) {$E$};
		\node [style=right label] (26) at (0.25, 4) {$C$};
		\node [style=right label] (27) at (1.75, 4) {$D$};
		\node [style=none] (28) at (-1.25, -0.25) {};
		\node [style=none] (29) at (1.5, -0.25) {};
	\end{pgfonlayer}
	\begin{pgfonlayer}{edgelayer}
		\draw (0) to (1);
		\draw (5) to (7);
		\draw (6) to (8);
		\draw [in=-90, out=180] (7) to (14.center);
		\draw [in=-90, out=0] (7) to (19.center);
		\draw (12) to (17.center);
		\draw (7) to (16.center);
		\draw [in=-90, out=-180] (8) to (15.center);
		\draw (18) to (21.center);
		\draw (15.center) to (20.center);
		\draw [in=-90, out=0] (8) to (13.center);
		\draw [in=-90, out=165] (1) to (28.center);
		\draw [in=-90, out=15] (1) to (29.center);
	\end{pgfonlayer}
\end{tikzpicture}}
\end{array} \]

In general, a Bayesian network (with possible latent variables), is a string diagram in \Stoch that (1) only has outputs and (2) consists only of copy maps and boxes which each have exactly one output.

By `a string diagram in \Stoch', we mean not only the stochastic matrix itself, but also its decomposition into components. We can formalise exactly what we mean by taking a perspective on Bayesian networks which draws inspiration from Lawvere's functorial semantics of algebraic theories \cite{Lawvere}. In this perspective, which elaborates on \cite[Ch. 4]{Fong12},
we maintain a conceptual distinction between the purely syntactic
object (the diagram) and its probabilistic interpretation.


Starting from a dag $G = (V_G,E_G)$, we construct a free CDU category $\free{G}$ which provides the syntax of causal structures labelled by $G$. The objects of $\free{G}$ are generated by the vertices of $G$, whereas the morphisms are generated by the following signature:
\[ \Sigma_G = \left\{ \; \scalebox{0.7}{\begin{tikzpicture}
	\begin{pgfonlayer}{nodelayer}
		\node [style=none] (2) at (0, 1.25) {};
		\node [style=none] (5) at (-1, -0.5) {};
		\node [style=none] (7) at (0, 0.5) {};
		\node [style=none] (8) at (-1, 0.5) {};
		\node [style=none] (9) at (1, -0.5) {};
		\node [style=none] (10) at (1, 0.5) {};
		\node [style=right label] (11) at (0.5, 1.25) {$A$};
		\node [style=none] (13) at (0, 0) {$a$};
		\node [style=none] (14) at (-0.75, -1.25) {};
		\node [style=none] (15) at (-0.75, -0.5) {};
		\node [style=none] (16) at (0.75, -1.25) {};
		\node [style=none] (17) at (0.75, -0.5) {};
		\node [style=right label] (18) at (-0.5, -1.25) {$B_1$};
		\node [style=right label] (19) at (1, -1.25) {$B_k$};
		\node [style=right label] (20) at (0, -0.75) {$\dots$};
	\end{pgfonlayer}
	\begin{pgfonlayer}{edgelayer}
		\draw (5.center) to (9.center);
		\draw (9.center) to (10.center);
		\draw (10.center) to (8.center);
		\draw (8.center) to (5.center);
		\draw (2.center) to (7.center);
		\draw (14.center) to (15.center);
		\draw (16.center) to (17.center);
	\end{pgfonlayer}
\end{tikzpicture}} 
   \ \  \middle| \ A \in V_G \text{ with parents } B_1, \dots, B_k \in V_G \right\}\]

\noindent Then $\free{G} := \freeCDU{V_G,\Sigma_G}$. \footnote{Note that $E_G$ is implicitly used in the construction of $\free{G}$: the edges of $G$ determine the parents of a vertex, and hence the input types of the symbols in $\Sigma_G$.} The following
result establishes that models (\emph{\`{a} la} Lawvere) of $\free{G}$
coincide with $G$-based Bayesian networks.



\begin{proposition}\label{prop:imaps}
There is a 1-1 correspondence between Bayesian networks based on the
dag $G$ and CDU functors of type $\free{G}\rightarrow \Stoch$.
\end{proposition}

\begin{proof} 
In one direction, consider a Bayesian network consisting of the dag
$G$ and, for each node $A \in V_G$, an assignment of a set $\tau(A)$
and a conditional probability $P(A | \textrm{Pa}(A))$. This data
yields a CDU functor $\mathcal F : \free{G}\rightarrow \Stoch$, defined by
the following mappings:
\[
\mathcal F \ \ ::\ \ 
\left\{
\begin{array}{rl}
A \in V_G & \ \ \mapsto\ \  \tau(A) \\
\scalebox{0.7}{\begin{tikzpicture}
	\begin{pgfonlayer}{nodelayer}
		\node [style=none] (2) at (0, 1.25) {};
		\node [style=none] (5) at (-1, -0.5) {};
		\node [style=none] (7) at (0, 0.5) {};
		\node [style=none] (8) at (-1, 0.5) {};
		\node [style=none] (9) at (1, -0.5) {};
		\node [style=none] (10) at (1, 0.5) {};
		\node [style=right label] (11) at (0.5, 1.25) {$A$};
		\node [style=none] (13) at (0, 0) {$a$};
		\node [style=none] (14) at (-0.75, -1.25) {};
		\node [style=none] (15) at (-0.75, -0.5) {};
		\node [style=none] (16) at (0.75, -1.25) {};
		\node [style=none] (17) at (0.75, -0.5) {};
		\node [style=right label] (18) at (-0.5, -1.25) {$B_1$};
		\node [style=right label] (19) at (1, -1.25) {$B_k$};
		\node [style=right label] (20) at (0, -0.75) {$\dots$};
	\end{pgfonlayer}
	\begin{pgfonlayer}{edgelayer}
		\draw (5.center) to (9.center);
		\draw (9.center) to (10.center);
		\draw (10.center) to (8.center);
		\draw (8.center) to (5.center);
		\draw (2.center) to (7.center);
		\draw (14.center) to (15.center);
		\draw (16.center) to (17.center);
	\end{pgfonlayer}
\end{tikzpicture}} & \ \ \mapsto\ \ 
\big( \mat f_{i_1...i_n}^j := P(A=j | \textrm{Pa}(A)=(i_1,\ldots,i_n)) \ \big)
\end{array}
\right.
\]
Conversely, let $\funF \colon \free{G}\rightarrow \Stoch$ be a CDU
functor. This defines a $G$-based Bayesian network by setting $\tau(A)
:= \funF(A)$ and $P(A=j | \textrm{Pa}(A)=(i_1,\ldots,i_n)) :=
\funF(a)_{i_1...i_n}^j$. It is
immediate that these two mappings are inverse to each other, thus
proving the statement. \QED
\end{proof}

This proposition
justifies the following definition of a category $\BN{G}$ of $G$-based
Bayesian networks: objects are CDU functors $\free{G}\rightarrow
\Stoch$ and arrows are monoidal natural transformations between them.


\section{Towards Causal Inference: the Smoking Scenario}\label{sec:smoking}

We will motivate our approach to causal inference via a classic
example, inspired by the one given in the Pearl's
book~\cite{PearlBook}. Imagine a dispute between a scientist and a
tobacco company. The scientist claims that smoking causes cancer. As a
source of evidence, the scientist cites a joint probability
distribution $\omega$ over variables $S$ for smoking and $C$ for
cancer, which disintegrates as in~\eqref{eq:naive-smoking} below, with
matrix $\mat c = \left(\begin{smallmatrix} 0.9 & 0.7 \\ 0.1 &   0.3
\end{smallmatrix}\right)$.  Inspecting this $\mat c : S \to C$, the
scientist notes that the probability of getting cancer for smokers
($0.3$) is three times as high as for non- smokers ($0.1$). Hence, the
scientist claims that smoking has a significant causal effect on
cancer.

\begin{wrapfigure}{r}{0pt}
\begin{minipage}{11em}
\begin{equation}\label{eq:naive-smoking}
\scalebox{0.8}{\begin{tikzpicture}
	\begin{pgfonlayer}{nodelayer}
		\node [style=small box] (1) at (3.5, 0.75) {$\mat c$};
		\node [style=none] (2) at (1.5, 0.75) {};
		\node [style=none] (4) at (0.5, -0.25) {$=$};
		\node [style=black dot] (7) at (2.5, -0.5) {};
		\node [style=small box] (8) at (2.5, -1.75) {$\mat s$};
		\node [style=none] (9) at (1.5, 2.25) {};
		\node [style=none] (10) at (3.5, 2.25) {};
		\node [style=medium box] (11) at (-1.25, -0.25) {$\mat \omega$};
		\node [style=none] (12) at (-1.75, 0.25) {};
		\node [style=none] (13) at (-1.75, 1.25) {};
		\node [style=none] (14) at (-0.75, 0.25) {};
		\node [style=none] (15) at (-0.75, 1.25) {};
		\node [style=right label] (16) at (-1.5, 1) {$S$};
		\node [style=right label] (17) at (1.75, 1.75) {$S$};
		\node [style=right label] (18) at (3.75, 1.75) {$C$};
		\node [style=right label] (19) at (-0.5, 1) {$C$};
	\end{pgfonlayer}
	\begin{pgfonlayer}{edgelayer}
		\draw (7) to (8);
		\draw [in=-90, out=15] (7) to (1);
		\draw [in=-90, out=165] (7) to (2.center);
		\draw (2.center) to (9.center);
		\draw (1) to (10.center);
		\draw (12.center) to (13.center);
		\draw (14.center) to (15.center);
	\end{pgfonlayer}
\end{tikzpicture}}
\end{equation}
\end{minipage}
\end{wrapfigure}
An important thing to stress here is that the scientist draws this
conclusion using not only the observational data $\mat \omega$ but also
from an assumed \textit{causal structure} which gave rise to that
data, as captured in the diagram in equation~\eqref{eq:naive-smoking}.
That is, rather than treating diagram \eqref{eq:naive-smoking} simply
as a calculation on the observational data, it can also be treated as
an assumption about the actual, physical mechanism that gave rise to
that data. Namely, this diagram encompasses the assumption
that there is some prior propensity for people to smoke captured
by $\mat s : I \to S$, which is both observed and fed into some other
process $\mat c : S \to C$ whereby an individuals choice to smoke
determines whether or not they get cancer.

\begin{wrapfigure}{r}{0pt}
\begin{minipage}{12em}
\begin{equation}
\label{eq:smoking-with-latent}
\scalebox{0.8}{\begin{tikzpicture}
	\begin{pgfonlayer}{nodelayer}
		\node [style=none] (1) at (2.75, 1) {};
		\node [style=none] (2) at (1.25, 1) {};
		\node [style=none] (4) at (0.25, 0) {$=$};
		\node [style=black dot] (7) at (2, 0.25) {};
		\node [style=small box] (8) at (2, -0.75) {$\mat s$};
		\node [style=none] (9) at (1.25, 3) {};
		\node [style=medium box] (11) at (-1.5, -0.25) {$\mat \omega$};
		\node [style=none] (12) at (-2, 0.25) {};
		\node [style=none] (13) at (-2, 1.25) {};
		\node [style=none] (14) at (-1, 0.25) {};
		\node [style=none] (15) at (-1, 1.25) {};
		\node [style=right label] (16) at (-1.75, 1) {$S$};
		\node [style=right label] (17) at (1.5, 2.5) {$S$};
		\node [style=right label] (18) at (3.5, 2.5) {$C$};
		\node [style=right label] (19) at (-0.75, 1) {$C$};
		\node [style=medium box] (20) at (3.25, 1.5) {$\mat c$};
		\node [style=none] (21) at (3.75, -1.25) {};
		\node [style=small box] (22) at (2.875, -3.25) {$\mat h$};
		\node [style=black dot] (23) at (2.875, -2.25) {};
		\node [style=none] (24) at (3.75, 1) {};
		\node [style=none] (25) at (3.25, 1.75) {};
		\node [style=none] (26) at (3.25, 3) {};
		\node [style=right label] (27) at (4, -0.5) {$H$};
		\node [style=none] (28) at (2, -1.25) {};
	\end{pgfonlayer}
	\begin{pgfonlayer}{edgelayer}
		\draw (7) to (8);
		\draw [in=-90, out=15] (7) to (1.center);
		\draw [in=-90, out=165] (7) to (2.center);
		\draw (2.center) to (9.center);
		\draw (12.center) to (13.center);
		\draw (14.center) to (15.center);
		\draw [in=-90, out=15] (23) to (21.center);
		\draw (21.center) to (24.center);
		\draw (25.center) to (26.center);
		\draw (22) to (23);
		\draw [in=-90, out=165] (23) to (28.center);
	\end{pgfonlayer}
\end{tikzpicture}}
\end{equation}
\end{minipage}
\end{wrapfigure}
The tobacco company, in turn, says that the scientists' assumptions
about the provenance of this data are too strong. While they
concede that \textit{in principle} it is possible for smoking to have
some influence on cancer, the scientist should allow for the
possibility that there is some latent common cause (e.g. genetic
conditions, stressful work environment, etc.) which leads people both
to smoke and get cancer. Hence, says the tobacco company, a `more
honest' causal structure to ascribe to the data $\omega$ is~\eqref{eq:smoking-with-latent}. This structure then allows for
either party to be correct. If the scientist is right, the output of
$\mat c : S \otimes H \to C$ depends mostly on its first input, i.e. the
causal path from smoking to cancer. If the tabacco company is right,
then $\mat c$ depends very little on its first input, and the correlation
between $S$ and $C$ can be explained almost entirely from the hidden
common cause.

So, who is right after all? Just from the observed distribution
$\mat \omega$, it is impossible to tell. So, the scientist proposes a
clinical trial, in which patients are randomly required to smoke or
not to smoke. We can model this situation by replacing $\mat s$ in
\eqref{eq:smoking-with-latent} with a process that ignores its inputs
and outputs the uniform state. Graphically, this looks like `cutting'
the link $\mat s$ between $H$ and $S$:
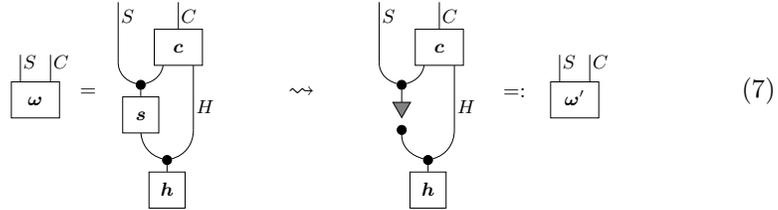
\begin{equation}
\label{eq:smoking-with-latent-cut}
\scalebox{0.8}{\begin{tikzpicture}
	\begin{pgfonlayer}{nodelayer}
		\node [style=none] (1) at (2.75, 1) {};
		\node [style=none] (2) at (1.25, 1) {};
		\node [style=none] (4) at (0.25, 0) {$=$};
		\node [style=black dot] (7) at (2, 0.25) {};
		\node [style=small box] (8) at (2, -0.75) {$\mat s$};
		\node [style=none] (9) at (1.25, 3) {};
		\node [style=medium box] (11) at (-1.5, -0.25) {$\mat \omega$};
		\node [style=none] (12) at (-2, 0.25) {};
		\node [style=none] (13) at (-2, 1.25) {};
		\node [style=none] (14) at (-1, 0.25) {};
		\node [style=none] (15) at (-1, 1.25) {};
		\node [style=right label] (16) at (-1.75, 1) {$S$};
		\node [style=right label] (17) at (1.5, 2.5) {$S$};
		\node [style=right label] (18) at (3.5, 2.5) {$C$};
		\node [style=right label] (19) at (-0.75, 1) {$C$};
		\node [style=medium box] (20) at (3.25, 1.5) {$\mat c$};
		\node [style=none] (21) at (3.75, -1.25) {};
		\node [style=small box] (22) at (2.875, -3.25) {$\mat h$};
		\node [style=black dot] (23) at (2.875, -2.25) {};
		\node [style=none] (24) at (3.75, 1) {};
		\node [style=none] (25) at (3.25, 1.75) {};
		\node [style=none] (26) at (3.25, 3) {};
		\node [style=right label] (27) at (4, -0.5) {$H$};
		\node [style=none] (28) at (2, -1.25) {};
	\end{pgfonlayer}
	\begin{pgfonlayer}{edgelayer}
		\draw (7) to (8);
		\draw [in=-90, out=15] (7) to (1.center);
		\draw [in=-90, out=165] (7) to (2.center);
		\draw (2.center) to (9.center);
		\draw (12.center) to (13.center);
		\draw (14.center) to (15.center);
		\draw [in=-90, out=15] (23) to (21.center);
		\draw (21.center) to (24.center);
		\draw (25.center) to (26.center);
		\draw (22) to (23);
		\draw [in=-90, out=165] (23) to (28.center);
	\end{pgfonlayer}
\end{tikzpicture}}
\qquad
\leadsto
\qquad
\scalebox{0.8}{\begin{tikzpicture}
	\begin{pgfonlayer}{nodelayer}
		\node [style=none] (1) at (-2.25, 1) {};
		\node [style=none] (2) at (-3.75, 1) {};
		\node [style=none] (4) at (0.75, 0) {$=:$};
		\node [style=black dot] (7) at (-3, 0.25) {};
		\node [style=uniform] (8) at (-3, -0.5) {};
		\node [style=none] (9) at (-3.75, 3) {};
		\node [style=medium box] (11) at (2.75, -0.25) {$\mat \omega'$};
		\node [style=none] (12) at (2.25, 0.25) {};
		\node [style=none] (13) at (2.25, 1.25) {};
		\node [style=none] (14) at (3.25, 0.25) {};
		\node [style=none] (15) at (3.25, 1.25) {};
		\node [style=right label] (16) at (2.5, 1) {$S$};
		\node [style=right label] (17) at (-3.5, 2.5) {$S$};
		\node [style=right label] (18) at (-1.5, 2.5) {$C$};
		\node [style=right label] (19) at (3.5, 1) {$C$};
		\node [style=medium box] (20) at (-1.75, 1.5) {$\mat c$};
		\node [style=none] (21) at (-1.25, -1.25) {};
		\node [style=small box] (22) at (-2.125, -3.25) {$\mat h$};
		\node [style=black dot] (23) at (-2.125, -2.25) {};
		\node [style=none] (24) at (-1.25, 1) {};
		\node [style=none] (25) at (-1.75, 1.75) {};
		\node [style=none] (26) at (-1.75, 3) {};
		\node [style=right label] (27) at (-1, -0.5) {$H$};
		\node [style=black dot] (28) at (-3, -1.25) {};
	\end{pgfonlayer}
	\begin{pgfonlayer}{edgelayer}
		\draw (7) to (8);
		\draw [in=-90, out=15] (7) to (1.center);
		\draw [in=-90, out=165] (7) to (2.center);
		\draw (2.center) to (9.center);
		\draw (12.center) to (13.center);
		\draw (14.center) to (15.center);
		\draw [in=-90, out=15] (23) to (21.center);
		\draw (21.center) to (24.center);
		\draw (25.center) to (26.center);
		\draw (22) to (23);
		\draw [in=-90, out=165] (23) to (28);
	\end{pgfonlayer}
\end{tikzpicture}}
\end{equation}
This captures the fact that variable $S$ is now
randomised and no longer dependent on any background factors.
This new distribution $\mat \omega'$ represents the data the scientist
would have obtained had they run the trial. That is, it gives the results
of an \textit{intervention} at $\mat s$. If this $\mat \omega'$ \textit{still}
shows a strong correlation between smoking and cancer, one can
conclude that smoking indeed causes cancer even when we assume the
weaker causal structure~\eqref{eq:smoking-with-latent}.

Unsurprisingly, the scientist fails to get ethical approval to run the
trial, and hence has only the observational data $\mat \omega$ to work
with. Given that the scientist only knows $\mat \omega$ (and not $\mat
c$ and $\mat h$), there is no way to compute $\mat \omega'$ in this
case. However, a key insight of statistical causal inference is that
sometimes it \textit{is} possible to compute interventional
distributions from observational ones. Continuing the smoking example,
suppose the scientist proposes the following revision to the causal
structure: they posit a structure \eqref{eq:smoking-hidden-2} that
includes a third observed variable (the presence of $T$ of tar in the
lungs), which completely mediates the causal effect of smoking on
cancer.

\begin{wrapfigure}{r}{0pt}
\begin{minipage}{15.5em}
\begin{equation}\label{eq:smoking-hidden-2}
\scalebox{0.8}{\begin{tikzpicture}
	\begin{pgfonlayer}{nodelayer}
		\node [style=large box] (0) at (0, 0) {$\mat \omega$};
		\node [style=none] (1) at (0, 2) {};
		\node [style=none] (2) at (1, 2) {};
		\node [style=none] (3) at (-1, 2) {};
		\node [style=none] (4) at (1, 0) {};
		\node [style=none] (5) at (-1, 0) {};
		\node [style=right label] (6) at (-0.75, 1.5) {$S$};
		\node [style=right label] (7) at (0.25, 1.5) {$T$};
		\node [style=right label] (8) at (1.25, 1.5) {$C$};
	\end{pgfonlayer}
	\begin{pgfonlayer}{edgelayer}
		\draw (1.center) to (0);
		\draw (2.center) to (4.center);
		\draw (3.center) to (5.center);
	\end{pgfonlayer}
\end{tikzpicture}} = \  
   \scalebox{0.8}{\begin{tikzpicture}
	\begin{pgfonlayer}{nodelayer}
		\node [style=none] (1) at (0.75, 0.25) {};
		\node [style=none] (2) at (-1.5, 0.5) {};
		\node [style=black dot] (7) at (0.75, -0.5) {};
		\node [style=small box] (8) at (0.75, -1.75) {$\mat s$};
		\node [style=none] (9) at (-1.5, 4.75) {};
		\node [style=right label] (17) at (-1.25, 4.25) {$S$};
		\node [style=right label] (18) at (1.875, 4.25) {$C$};
		\node [style=semilarge box] (20) at (1.625, 3.25) {$\mat c$};
		\node [style=none] (21) at (2.5, -2.25) {};
		\node [style=small box] (22) at (1.625, -4) {$\mat h$};
		\node [style=black dot] (23) at (1.625, -3) {};
		\node [style=none] (24) at (2.5, 3) {};
		\node [style=none] (25) at (1.625, 3.5) {};
		\node [style=none] (26) at (1.625, 4.75) {};
		\node [style=right label] (27) at (2.75, -0.25) {$H$};
		\node [style=none] (28) at (0.75, -2.25) {};
		\node [style=small box] (29) at (0.75, 0.75) {$\mat t$};
		\node [style=none] (30) at (0.75, 3) {};
		\node [style=none] (31) at (-0.25, 2.75) {};
		\node [style=black dot] (32) at (0.75, 2) {};
		\node [style=none] (33) at (-0.25, 4.75) {};
		\node [style=right label] (34) at (0, 4.25) {$T$};
	\end{pgfonlayer}
	\begin{pgfonlayer}{edgelayer}
		\draw (7) to (8);
		\draw [in=-90, out=180, looseness=0.75] (7) to (2.center);
		\draw (2.center) to (9.center);
		\draw [in=-90, out=15] (23) to (21.center);
		\draw (21.center) to (24.center);
		\draw (25.center) to (26.center);
		\draw (22) to (23);
		\draw [in=-90, out=165] (23) to (28.center);
		\draw [in=-90, out=-180] (32) to (31.center);
		\draw (31.center) to (33.center);
		\draw (29) to (32);
		\draw (7) to (1.center);
		\draw (32) to (30.center);
	\end{pgfonlayer}
\end{tikzpicture}}
\end{equation}
\end{minipage}
\end{wrapfigure}
As with our simpler structure, the diagram \eqref{eq:smoking-hidden-2}
contains some assumptions about the provenance of the data
$\mat \omega$. In particular, by omitting wires, we are asserting there is
no \textit{direct} causal link between certain variables. The
absence of an $H$-labelled input to $\mat t$ says there is no direct causal
link from $H$ to $T$ (only mediated by $S$), and the absence of an
$S$-labelled input wire into $\mat c$ captures that there is no direct
causal link from $S$ to $C$ (only mediated by $T$). In the traditional
approach to causal inference, such relationships are typically
captured by a graph-theoretic property called \textit{d-separation} on
the dag associated with the causal structure.

We can again imagine intervening at $S$ by replacing $\mat s : H \to S$ by
$\uniform \circ \counit$. Again, this `cutting' of the diagram will
result in a new interventional distribution $\mat \omega'$. However, unike
before, it \textit{is} possible to compute this distribution from the
observational distribution $\mat\omega$.

However, in order to do that, we first need to develop the appropriate categorical framework. In Section \ref{sec:surgery}, we will model `cutting' as a functor. In \ref{sec:comb}, we will introduce a generalisation of disintegration, which we call \textit{comb
  disintegration}. These tools will enable us to compute $\mat \omega'$ for $\mat \omega$, in Section \ref{sec:smoking-return}.

\section{Interventional Distributions as Diagram Surgery}\label{sec:surgery}



The goal of this section is to define the `cut' operation
in~\eqref{eq:smoking-with-latent-cut} as an endofunctor on the category of Bayesian
networks. First, we observe that such an operation exclusively concerns the
string diagram part of a Bayesian network: following the functorial semantics
given in Section~\ref{sec:bn}, it is thus appropriate to define cut as an
endofunctor on $\free{G}$, for a given dag $G$.

\begin{definition} 
For a fixed node $A \in V_G$ in a graph $G$, let $\syncut{A} \colon
\free{G} \to \free{G}$ be the CDU
functor freely obtained by the following action on the generators
$(V_G, \Sigma_G)$ of $\free{G}$:
\begin{itemize}
\item For each object $B \in V_G$, $\syncut{A}(B) = B$.

\item $\syncut{A}(\,\scalebox{0.6}{\begin{tikzpicture}
	\begin{pgfonlayer}{nodelayer}
		\node [style=none] (2) at (0, 1.25) {};
		\node [style=none] (5) at (-1, -0.5) {};
		\node [style=none] (7) at (0, 0.5) {};
		\node [style=none] (8) at (-1, 0.5) {};
		\node [style=none] (9) at (1, -0.5) {};
		\node [style=none] (10) at (1, 0.5) {};
		\node [style=right label] (11) at (0.5, 1.25) {$A$};
		\node [style=none] (13) at (0, 0) {$a$};
		\node [style=none] (14) at (-0.75, -1.25) {};
		\node [style=none] (15) at (-0.75, -0.5) {};
		\node [style=none] (16) at (0.75, -1.25) {};
		\node [style=none] (17) at (0.75, -0.5) {};
		\node [style=right label] (18) at (-0.5, -1.25) {$B_1$};
		\node [style=right label] (19) at (1, -1.25) {$B_k$};
		\node [style=right label] (20) at (0, -0.75) {$\dots$};
	\end{pgfonlayer}
	\begin{pgfonlayer}{edgelayer}
		\draw (5.center) to (9.center);
		\draw (9.center) to (10.center);
		\draw (10.center) to (8.center);
		\draw (8.center) to (5.center);
		\draw (2.center) to (7.center);
		\draw (14.center) to (15.center);
		\draw (16.center) to (17.center);
	\end{pgfonlayer}
\end{tikzpicture}}\hspace*{-0.8em})
  = \scalebox{0.6}{\begin{tikzpicture}
	\begin{pgfonlayer}{nodelayer}
		\node [style=none] (2) at (0, 1.25) {};
		\node [style=none] (7) at (0, 0.5) {};
		\node [style=right label] (11) at (0.5, 1.25) {$A$};
		\node [style=none] (14) at (-0.75, -1.25) {};
		\node [style=none] (15) at (-0.75, -0.5) {};
		\node [style=none] (16) at (0.75, -1.25) {};
		\node [style=none] (17) at (0.75, -0.5) {};
		\node [style=right label] (18) at (-0.5, -1.25) {$B_1$};
		\node [style=right label] (19) at (1, -1.25) {$B_k$};
		\node [style=right label] (20) at (0, -0.75) {$\dots$};
		\node [style=black dot] (27) at (0.75, -0.5) {};
		\node [style=uniform] (30) at (0, 0.5) {};
		\node [style=black dot] (33) at (-0.75, -0.5) {};
	\end{pgfonlayer}
	\begin{pgfonlayer}{edgelayer}
		\draw (2.center) to (7.center);
		\draw (14.center) to (15.center);
		\draw (16.center) to (17.center);
	\end{pgfonlayer}
\end{tikzpicture}}\!$ and
  $\syncut{A}(\,\scalebox{0.6}{\begin{tikzpicture}
	\begin{pgfonlayer}{nodelayer}
		\node [style=none] (2) at (0, 1.25) {};
		\node [style=none] (5) at (-1, -0.5) {};
		\node [style=none] (7) at (0, 0.5) {};
		\node [style=none] (8) at (-1, 0.5) {};
		\node [style=none] (9) at (1, -0.5) {};
		\node [style=none] (10) at (1, 0.5) {};
		\node [style=right label] (11) at (0.5, 1.25) {$B$};
		\node [style=none] (13) at (0, 0) {$b$};
		\node [style=none] (14) at (-0.75, -1.25) {};
		\node [style=none] (15) at (-0.75, -0.5) {};
		\node [style=none] (16) at (0.75, -1.25) {};
		\node [style=none] (17) at (0.75, -0.5) {};
		\node [style=right label] (18) at (-0.5, -1.25) {$C_1$};
		\node [style=right label] (19) at (1, -1.25) {$C_j$};
		\node [style=right label] (20) at (0, -0.75) {$\dots$};
	\end{pgfonlayer}
	\begin{pgfonlayer}{edgelayer}
		\draw (5.center) to (9.center);
		\draw (9.center) to (10.center);
		\draw (10.center) to (8.center);
		\draw (8.center) to (5.center);
		\draw (2.center) to (7.center);
		\draw (14.center) to (15.center);
		\draw (16.center) to (17.center);
	\end{pgfonlayer}
\end{tikzpicture}}\hspace*{-0.8em}) =
  \scalebox{0.6}{\begin{tikzpicture}
	\begin{pgfonlayer}{nodelayer}
		\node [style=none] (2) at (0, 1.25) {};
		\node [style=none] (5) at (-1, -0.5) {};
		\node [style=none] (7) at (0, 0.5) {};
		\node [style=none] (8) at (-1, 0.5) {};
		\node [style=none] (9) at (1, -0.5) {};
		\node [style=none] (10) at (1, 0.5) {};
		\node [style=right label] (11) at (0.5, 1.25) {$B$};
		\node [style=none] (13) at (0, 0) {$b$};
		\node [style=none] (14) at (-0.75, -1.25) {};
		\node [style=none] (15) at (-0.75, -0.5) {};
		\node [style=none] (16) at (0.75, -1.25) {};
		\node [style=none] (17) at (0.75, -0.5) {};
		\node [style=right label] (18) at (-0.5, -1.25) {$C_1$};
		\node [style=right label] (19) at (1, -1.25) {$C_j$};
		\node [style=right label] (20) at (0, -0.75) {$\dots$};
	\end{pgfonlayer}
	\begin{pgfonlayer}{edgelayer}
		\draw (5.center) to (9.center);
		\draw (9.center) to (10.center);
		\draw (10.center) to (8.center);
		\draw (8.center) to (5.center);
		\draw (2.center) to (7.center);
		\draw (14.center) to (15.center);
		\draw (16.center) to (17.center);
	\end{pgfonlayer}
\end{tikzpicture}}$ for any other
  $\scalebox{0.6}{\begin{tikzpicture}
	\begin{pgfonlayer}{nodelayer}
		\node [style=none] (2) at (0, 1.25) {};
		\node [style=none] (5) at (-1, -0.5) {};
		\node [style=none] (7) at (0, 0.5) {};
		\node [style=none] (8) at (-1, 0.5) {};
		\node [style=none] (9) at (1, -0.5) {};
		\node [style=none] (10) at (1, 0.5) {};
		\node [style=right label] (11) at (0.5, 1.25) {$B$};
		\node [style=none] (13) at (0, 0) {$b$};
		\node [style=none] (14) at (-0.75, -1.25) {};
		\node [style=none] (15) at (-0.75, -0.5) {};
		\node [style=none] (16) at (0.75, -1.25) {};
		\node [style=none] (17) at (0.75, -0.5) {};
		\node [style=right label] (18) at (-0.5, -1.25) {$C_1$};
		\node [style=right label] (19) at (1, -1.25) {$C_j$};
		\node [style=right label] (20) at (0, -0.75) {$\dots$};
	\end{pgfonlayer}
	\begin{pgfonlayer}{edgelayer}
		\draw (5.center) to (9.center);
		\draw (9.center) to (10.center);
		\draw (10.center) to (8.center);
		\draw (8.center) to (5.center);
		\draw (2.center) to (7.center);
		\draw (14.center) to (15.center);
		\draw (16.center) to (17.center);
	\end{pgfonlayer}
\end{tikzpicture}\!}\!\!  \in \Sigma_G$.
 \end{itemize}  
 \end{definition}

Intuitively, $\syncut{A}$ applied to a string diagram $f$ of
$\free{G}$ removes from $f$ each occurrence of a box with output wire
of type $A$.

Proposition~\ref{prop:imaps} allows us to ``transport'' the cutting
operation over to Bayesian networks. Given any Bayesian network based
on $G$, let $\funF \colon \free{G} \to \Stoch$ be the corresponding
CDU functor given by Proposition~\ref{prop:imaps}. Then, we can define
its $A$-cutting as the Bayesian network identified by the CDU functor
$\funF \circ \syncut{A}$. This yields an (idempotent) endofunctor
$\cut{A} \colon \BN{G} \to \BN{G}$.


\section{The Comb Factorisation}\label{sec:comb}

Thanks to the developments of Section~\ref{sec:surgery}, we can
understand the transition from left to right
in~\eqref{eq:smoking-with-latent-cut} as the application of the
functor $\cut{S}$ applied to the `Smoking' node $S$. The next
step is being able to actually compute the individual
$\Stoch$-morphisms appearing in~\eqref{eq:smoking-hidden-2}, to give
an answer to the causality question.

\begin{wrapfigure}{r}{0pt}
$\;\scalebox{0.7}{\begin{tikzpicture}
	\begin{pgfonlayer}{nodelayer}
		\node [style=none] (0) at (-5.5, -1.25) {};
		\node [style=none] (1) at (0, -1.25) {};
		\node [style=none] (2) at (-4, 0) {};
		\node [style=none] (3) at (-2.5, 0) {};
		\node [style=none] (4) at (-1.25, 0) {$=$};
		\node [style=none] (5) at (-5.5, 0) {};
		\node [style=none] (6) at (-4, 0) {};
		\node [style=none] (7) at (-2.5, 1.25) {};
		\node [style=none] (8) at (0, 1.25) {};
		\node [style=none] (13) at (5.5, -1.25) {};
		\node [style=none] (14) at (4, 0) {};
		\node [style=none] (15) at (2.5, 0) {};
		\node [style=none] (16) at (5.5, 0) {};
		\node [style=none] (17) at (4, 0) {};
		\node [style=none] (18) at (2.5, 1.25) {};
		\node [style=none] (22) at (1.25, 0) {$=$};
		\node [style=none] (24) at (4.25, 0.25) {};
	\end{pgfonlayer}
	\begin{pgfonlayer}{edgelayer}
		\draw [in=-90, out=-90, looseness=1.75] (2.center) to (3.center);
		\draw (3.center) to (7.center);
		\draw [in=90, out=90, looseness=1.75] (5.center) to (6.center);
		\draw (0.center) to (5.center);
		\draw [style=swap] (8.center) to (1.center);
		\draw [in=-90, out=-90, looseness=1.75] (14.center) to (15.center);
		\draw (15.center) to (18.center);
		\draw [in=90, out=90, looseness=1.75] (16.center) to (17.center);
		\draw (13.center) to (16.center);
	\end{pgfonlayer}
\end{tikzpicture}}\;$
\end{wrapfigure}
In order to do that, we want to work in a setting where $\mat t \colon S
\to T$ can be isolated and `extracted' from
\eqref{eq:smoking-hidden-2}. What is left behind is a stochastic
matrix with a `hole' where $\mat t$ has been extracted. To define `morphisms
with holes', it is convenient to pass from SMCs to compact closed categories
(see e.g.~\cite{Sel2009-graphical}). $\Stoch$ is not itself compact closed, but it
embeds into $\MatRp$, whose morphisms are \textit{all} matrices over
positive numbers. $\MatRp$ has a (self-dual) compact closed structure;
that means, for any set $A$ there is a `cap' $\cap \colon A \tns A \to
I$ and a `cup' $\cup \colon I \to A \tns A$, which satisfy the `yanking'
equations on the right.  As matrices, caps and cups are defined by
$\cap_{ij} = \cup^{ij} = \delta_i^j$. Intuitively, they amount to
`bent' identity wires. Another aspect of $\MatRp$ that is useful to
recall is the following handy characterisation of the subcategory
$\Stoch$.

\begin{lemma} \label{prop:Stoch_morphisms_in_Mat}
A morphism $\mat f \colon A \to B$ in $\MatRp$ is a stochastic matrix (thus
a morphism of $\Stoch$) if and only if \eqref{eq:discarding-final}
holds.
\end{lemma}

A suitable notion of `stochastic map with a hole' is provided
by a \emph{comb}. These structures originate in the study of certain
kinds of quantum channels~\cite{comb}.


\begin{definition}
\label{def:comb}
A $2$-comb in $\Stoch$ is a morphism $\mat f \colon A_1 \otimes A_2
\to B_1 \otimes B_2$ satisfying, for some other morphism $\mat f'
\colon A_1 \to B_1$,
\begin{equation}\label{eq:2-comb}
\scalebox{0.8}{\begin{tikzpicture}
	\begin{pgfonlayer}{nodelayer}
		\node [style=semilarge box] (0) at (-2.25, 0) {$\mat f$};
		\node [style=none] (1) at (-3, -1.5) {};
		\node [style=none] (2) at (-1.5, -1.5) {};
		\node [style=none] (3) at (-3, 1.5) {};
		\node [style=black dot] (4) at (-1.5, 1.5) {};
		\node [style=none] (5) at (0.25, 0) {$=$};
		\node [style=small box] (6) at (2, 0) {$\mat f'$};
		\node [style=none] (7) at (2, -1.5) {};
		\node [style=none] (8) at (3.25, -1.5) {};
		\node [style=none] (9) at (2, 1.5) {};
		\node [style=black dot] (10) at (3.25, 0) {};
		\node [style=left label] (11) at (-3.25, -1.25) {$A_1$};
		\node [style=left label] (12) at (-3.25, 1) {$B_1$};
		\node [style=right label] (13) at (-1.25, -1.25) {$A_2$};
		\node [style=right label] (14) at (-1.25, 1) {$B_2$};
		\node [style=left label] (15) at (1.75, -1.25) {$A_1$};
		\node [style=left label] (16) at (1.75, 1.25) {$B_1$};
		\node [style=right label] (17) at (3.5, -1.25) {$A_2$};
	\end{pgfonlayer}
	\begin{pgfonlayer}{edgelayer}
		\draw (1.center) to (3.center);
		\draw (2.center) to (4);
		\draw (7.center) to (9.center);
		\draw (8.center) to (10);
	\end{pgfonlayer}
\end{tikzpicture}}
\end{equation}
\end{definition}

This definition extends inductively to \textit{$n$-combs}, where we
require that discarding the rightmost output yields $\mat f' \otimes
\counit$, for some $(n-1)$-comb $\mat f'$. However, for our purposes,
restricting to 2-combs will suffice.

The intuition behind condition \eqref{eq:2-comb} is that the
contribution from input $A_2$ is only visible via output $B_2$. Thus,
if we discard $B_2$ we may as well discard $A_2$. In other words, the
input/output pair $A_2,B_2$ happen `after' the pair $A_1, B_1$. Hence,
it is typical to depict 2-combs in the shape of a (hair) comb, with 2
`teeth', as in \eqref{eq:2ord-comb} below:

\noindent\begin{minipage}{.4\linewidth}
\begin{equation}\label{eq:2ord-comb}
  \scalebox{0.8}{\begin{tikzpicture}
	\begin{pgfonlayer}{nodelayer}
		\node [style=semilarge box] (0) at (0, 0) {$\mat f$};
		\node [style=none] (1) at (0.75, -0.5) {};
		\node [style=none] (2) at (0.75, -1.5) {};
		\node [style=right label] (3) at (1, -1) {$A_2$};
		\node [style=none] (4) at (-0.75, 1.5) {};
		\node [style=none] (5) at (-0.75, 0.5) {};
		\node [style=right label] (6) at (-0.5, 1) {$B_1$};
		\node [style=none] (7) at (-0.75, -0.5) {};
		\node [style=none] (8) at (-0.75, -1.5) {};
		\node [style=right label] (9) at (-0.5, -1) {$A_1$};
		\node [style=none] (10) at (0.75, 1.5) {};
		\node [style=none] (11) at (0.75, 0.5) {};
		\node [style=right label] (12) at (1, 1) {$B_2$};
	\end{pgfonlayer}
	\begin{pgfonlayer}{edgelayer}
		\draw (1.center) to (2.center);
		\draw (4.center) to (5.center);
		\draw (7.center) to (8.center);
		\draw (10.center) to (11.center);
	\end{pgfonlayer}
\end{tikzpicture}} \leadsto\;\;
   \scalebox{0.8}{\begin{tikzpicture}
	\begin{pgfonlayer}{nodelayer}
		\node [style=none] (0) at (-1, 2.5) {};
		\node [style=none] (1) at (-1, 1.5) {};
		\node [style=none] (2) at (0.5, 1.5) {};
		\node [style=none] (3) at (0.5, -1.5) {};
		\node [style=none] (4) at (-1, -1.5) {};
		\node [style=none] (5) at (-1, -2.5) {};
		\node [style=none] (6) at (1.5, -2.5) {};
		\node [style=none] (7) at (1.5, 2.5) {};
		\node [style=none] (8) at (1, 0) {$\mat f$};
		\node [style=none] (9) at (-0.25, -1) {};
		\node [style=none] (10) at (-0.25, -1.5) {};
		\node [style=none] (11) at (-0.25, -2.5) {};
		\node [style=none] (12) at (-0.25, -3) {};
		\node [style=none] (13) at (-0.25, 3) {};
		\node [style=none] (14) at (-0.25, 2.5) {};
		\node [style=none] (15) at (-0.25, 1.5) {};
		\node [style=none] (16) at (-0.25, 1) {};
		\node [style=left label] (17) at (-0.5, -3) {$A_1$};
		\node [style=left label] (18) at (-0.5, -1) {$B_1$};
		\node [style=left label] (19) at (-0.5, 1) {$A_2$};
		\node [style=left label] (20) at (-0.5, 3) {$B_2$};
	\end{pgfonlayer}
	\begin{pgfonlayer}{edgelayer}
		\draw (0.center) to (1.center);
		\draw (1.center) to (2.center);
		\draw (2.center) to (3.center);
		\draw (3.center) to (4.center);
		\draw (4.center) to (5.center);
		\draw (5.center) to (6.center);
		\draw (6.center) to (7.center);
		\draw (7.center) to (0.center);
		\draw (9.center) to (10.center);
		\draw (11.center) to (12.center);
		\draw (13.center) to (14.center);
		\draw (15.center) to (16.center);
	\end{pgfonlayer}
\end{tikzpicture}}
\end{equation}
\end{minipage}%
\begin{minipage}{.6\linewidth}
\begin{equation}\label{eq:2-comb-compose}
  \scalebox{0.8}{\begin{tikzpicture}
	\begin{pgfonlayer}{nodelayer}
		\node [style=none] (0) at (-4.25, 2.5) {};
		\node [style=none] (1) at (-4.25, 1.5) {};
		\node [style=none] (2) at (-2.5, 1.5) {};
		\node [style=none] (3) at (-2.5, -1.5) {};
		\node [style=none] (4) at (-4.25, -1.5) {};
		\node [style=none] (5) at (-4.25, -2.5) {};
		\node [style=none] (6) at (-1.5, -2.5) {};
		\node [style=none] (7) at (-1.5, 2.5) {};
		\node [style=none] (8) at (-2, 0) {$\mat f$};
		\node [style=none] (9) at (-3.5, -0.5) {};
		\node [style=none] (10) at (-3.5, -1.5) {};
		\node [style=none] (11) at (-3.5, -2.5) {};
		\node [style=none] (12) at (-3.5, -3) {};
		\node [style=none] (13) at (-3.5, 3) {};
		\node [style=none] (14) at (-3.5, 2.5) {};
		\node [style=none] (15) at (-3.5, 1.5) {};
		\node [style=none] (16) at (-3.5, 0.5) {};
		\node [style=left label] (17) at (-3.75, -3) {$A_1$};
		\node [style=left label] (18) at (-3.75, -1) {$B_1$};
		\node [style=left label] (19) at (-3.75, 1) {$A_2$};
		\node [style=left label] (20) at (-3.75, 3) {$B_2$};
		\node [style=small box] (21) at (-3.5, 0) {$\mat g$};
		\node [style=none] (22) at (0, 0) {$:=$};
		\node [style=semilarge box] (23) at (3, -1.5) {$\mat f$};
		\node [style=none] (24) at (3.75, -2) {};
		\node [style=none] (25) at (3.75, -2.25) {};
		\node [style=left label] (26) at (3.5, -2.5) {$A_2$};
		\node [style=none] (27) at (2.25, 0) {};
		\node [style=none] (28) at (2.25, -1) {};
		\node [style=none] (29) at (5.25, 1) {};
		\node [style=none] (30) at (2.25, 1) {};
		\node [style=left label] (31) at (2, -0.5) {$B_1$};
		\node [style=left label] (32) at (2, 1.5) {$A_2$};
		\node [style=small box] (33) at (2.25, 0.5) {$\mat g$};
		\node [style=none] (34) at (2.25, -2) {};
		\node [style=none] (35) at (2.25, -3) {};
		\node [style=left label] (36) at (2, -2.5) {$A_1$};
		\node [style=none] (37) at (3.75, 3) {};
		\node [style=none] (38) at (3.75, -1) {};
		\node [style=left label] (39) at (3.5, 2.5) {$B_2$};
		\node [style=none] (40) at (5.25, -2.25) {};
	\end{pgfonlayer}
	\begin{pgfonlayer}{edgelayer}
		\draw (0.center) to (1.center);
		\draw (1.center) to (2.center);
		\draw (2.center) to (3.center);
		\draw (3.center) to (4.center);
		\draw (4.center) to (5.center);
		\draw (5.center) to (6.center);
		\draw (6.center) to (7.center);
		\draw (7.center) to (0.center);
		\draw (9.center) to (10.center);
		\draw (11.center) to (12.center);
		\draw (13.center) to (14.center);
		\draw (15.center) to (16.center);
		\draw (24.center) to (25.center);
		\draw (27.center) to (28.center);
		\draw [in=90, out=90] (29.center) to (30.center);
		\draw (34.center) to (35.center);
		\draw (37.center) to (38.center);
		\draw [in=-90, out=-90, looseness=1.50] (40.center) to (25.center);
		\draw (40.center) to (29.center);
	\end{pgfonlayer}
\end{tikzpicture}}
\end{equation}
\end{minipage}

\noindent While combs themselves live in \Stoch, $\MatRp$ accommodates
a second-order reading of the transition $\leadsto$ in~\eqref{eq:2ord-comb}:
we can treat $\mat f$ as a map which expects as input a
map $\mat g \colon B_1 \to A_2$ and produces as output a map of type $A_1 \to
B_2$. Plugging $\mat g \colon B_1 \to A_2$ into the 2-comb can be
formally defined in $\MatRp$ by composing $\mat f$ and $\mat g$ in the
usual way, then feeding the output of $\mat g$ into the second input
of $\mat f$, using caps and cups, as in~\eqref{eq:2-comb-compose}.

Importantly, for generic $\mat f$ and $\mat g$ of $\Stoch$, there is no
guarantee that forming the composite~\eqref{eq:2-comb-compose}
in $\MatRp$ yields a valid $\Stoch$-morphism, i.e. a morphism satisfying the
finality equation~\eqref{eq:discarding-final}. However, if $\mat f$ is a
2-comb and $\mat g$ is a $\Stoch$-morphism, equation \eqref{eq:2-comb} enables
a discarding map plugged into the output $B_2$ in \eqref{eq:2-comb-compose} to
`fall through' the right side of $\mat f$, which guarantees that the composed
map satisfies the finality equation for discarding. Hence the composition indeed yields a \Stoch-morphism:
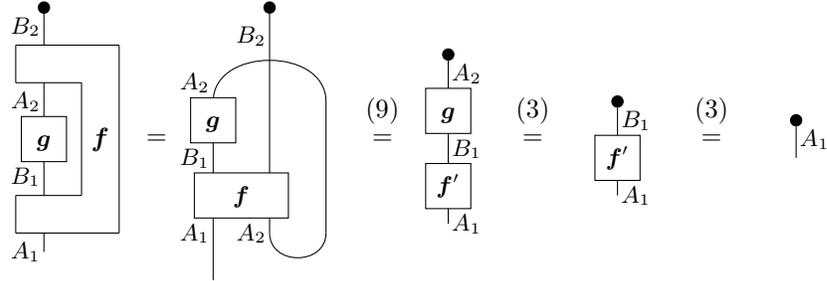
\[
\qquad\begin{tikzpicture}
	\begin{pgfonlayer}{nodelayer}
		\node [style=none] (0) at (-11, 2.5) {};
		\node [style=none] (1) at (-11, 1.5) {};
		\node [style=none] (2) at (-9.25, 1.5) {};
		\node [style=none] (3) at (-9.25, -1.5) {};
		\node [style=none] (4) at (-11, -1.5) {};
		\node [style=none] (5) at (-11, -2.5) {};
		\node [style=none] (6) at (-8.25, -2.5) {};
		\node [style=none] (7) at (-8.25, 2.5) {};
		\node [style=none] (8) at (-8.75, 0) {$\mat f$};
		\node [style=none] (9) at (-10.25, -0.5) {};
		\node [style=none] (10) at (-10.25, -1.5) {};
		\node [style=none] (11) at (-10.25, -2.5) {};
		\node [style=none] (12) at (-10.25, -3) {};
		\node [style=black dot] (13) at (-10.25, 3.5) {};
		\node [style=none] (14) at (-10.25, 2.5) {};
		\node [style=none] (15) at (-10.25, 1.5) {};
		\node [style=none] (16) at (-10.25, 0.5) {};
		\node [style=left label] (17) at (-10.5, -3) {$A_1$};
		\node [style=left label] (18) at (-10.5, -1) {$B_1$};
		\node [style=left label] (19) at (-10.5, 1) {$A_2$};
		\node [style=left label] (20) at (-10.5, 3) {$B_2$};
		\node [style=small box] (21) at (-10.25, 0) {$\mat g$};
		\node [style=none] (22) at (-7.25, 0) {$=$};
		\node [style=semilarge box] (23) at (-5, -1.5) {$\mat f$};
		\node [style=none] (24) at (-4.25, -2) {};
		\node [style=none] (25) at (-4.25, -2.5) {};
		\node [style=left label] (26) at (-4.5, -2.5) {$A_2$};
		\node [style=none] (27) at (-5.75, 0) {};
		\node [style=none] (28) at (-5.75, -1) {};
		\node [style=none] (29) at (-2.75, 1) {};
		\node [style=none] (30) at (-5.75, 1) {};
		\node [style=left label] (31) at (-6, -0.5) {$B_1$};
		\node [style=left label] (32) at (-6, 1.5) {$A_2$};
		\node [style=small box] (33) at (-5.75, 0.5) {$\mat g$};
		\node [style=none] (34) at (-5.75, -2) {};
		\node [style=none] (35) at (-5.75, -3.75) {};
		\node [style=left label] (36) at (-6, -2.5) {$A_1$};
		\node [style=black dot] (37) at (-4.25, 3.5) {};
		\node [style=none] (38) at (-4.25, -1) {};
		\node [style=left label] (39) at (-4.5, 2.75) {$B_2$};
		\node [style=none] (40) at (-2.75, -2.5) {};
		\node [style=none] (41) at (-1.25, 0) {$=$};
		\node [style=small box] (42) at (0.5, -1.25) {$\mat f'$};
		\node [style=none] (43) at (0.5, 0.25) {};
		\node [style=none] (44) at (0.5, -0.75) {};
		\node [style=none] (45) at (0.5, -1.75) {};
		\node [style=none] (46) at (0.5, -2.25) {};
		\node [style=black dot] (47) at (0.5, 2.25) {};
		\node [style=none] (48) at (0.5, 1.25) {};
		\node [style=right label] (49) at (0.75, -2.25) {$A_1$};
		\node [style=right label] (50) at (0.75, -0.25) {$B_1$};
		\node [style=right label] (51) at (0.75, 1.75) {$A_2$};
		\node [style=small box] (52) at (0.5, 0.75) {$\mat g$};
		\node [style=none] (53) at (2.75, 0) {$=$};
		\node [style=small box] (54) at (5, -0.5) {$\mat f'$};
		\node [style=none] (55) at (5, -1) {};
		\node [style=none] (56) at (5, -1.5) {};
		\node [style=black dot] (57) at (5, 1) {};
		\node [style=none] (58) at (5, 0) {};
		\node [style=right label] (59) at (5.25, -1.5) {$A_1$};
		\node [style=right label] (60) at (5.25, 0.5) {$B_1$};
		\node [style=none] (61) at (7.5, 0) {$=$};
		\node [style=black dot] (62) at (9.75, 0.5) {};
		\node [style=none] (63) at (9.75, -0.5) {};
		\node [style=right label] (64) at (10, 0) {$A_1$};
		\node [style=none] (65) at (-1.25, 0.75) {\eqref{eq:2-comb}};
		\node [style=none] (66) at (7.5, 0.75) {\eqref{eq:discarding-final}};
		\node [style=none] (67) at (2.75, 0.75) {\eqref{eq:discarding-final}};
	\end{pgfonlayer}
	\begin{pgfonlayer}{edgelayer}
		\draw (0.center) to (1.center);
		\draw (1.center) to (2.center);
		\draw (2.center) to (3.center);
		\draw (3.center) to (4.center);
		\draw (4.center) to (5.center);
		\draw (5.center) to (6.center);
		\draw (6.center) to (7.center);
		\draw (7.center) to (0.center);
		\draw (9.center) to (10.center);
		\draw (11.center) to (12.center);
		\draw (13) to (14.center);
		\draw (15.center) to (16.center);
		\draw (24.center) to (25.center);
		\draw (27.center) to (28.center);
		\draw [in=90, out=90, looseness=1.25] (29.center) to (30.center);
		\draw (34.center) to (35.center);
		\draw (37) to (38.center);
		\draw [in=-90, out=-90, looseness=1.50] (40.center) to (25.center);
		\draw (40.center) to (29.center);
		\draw (43.center) to (44.center);
		\draw (45.center) to (46.center);
		\draw (47) to (48.center);
		\draw (55.center) to (56.center);
		\draw (57) to (58.center);
		\draw (62) to (63.center);
	\end{pgfonlayer}
\end{tikzpicture}
\]



With the concept of 2-combs in hand, we can state our factorisation result.

\begin{theorem} 
\label{thm:combuniqueextraction}
For any state $\mat \omega \colon I \to A \otimes B \otimes C$ of
$\Stoch$ with full support, there exists a unique 2-comb $\mat f : B
\to A \otimes C$ and stochastic matrix $\mat g \colon A \to B$ such
that, in $\MatRp$:
\vspace*{-1em}
\begin{equation}\label{eq:comb-disint}
\scalebox{0.8}{\begin{tikzpicture}
	\begin{pgfonlayer}{nodelayer}
		\node [style=none] (0) at (1.25, 2.25) {};
		\node [style=none] (1) at (1.25, 1.25) {};
		\node [style=none] (2) at (3, 1.25) {};
		\node [style=none] (3) at (3, -1.75) {};
		\node [style=none] (4) at (1.25, -1.75) {};
		\node [style=none] (5) at (1.25, -2.75) {};
		\node [style=none] (6) at (4, -2.75) {};
		\node [style=none] (7) at (4, 2.25) {};
		\node [style=none] (8) at (3.5, -0.25) {$\mat f$};
		\node [style=black dot] (9) at (2, -1.25) {};
		\node [style=none] (10) at (2, -1.75) {};
		\node [style=none] (13) at (2, 3) {};
		\node [style=none] (14) at (2, 2.25) {};
		\node [style=none] (15) at (2, 1.25) {};
		\node [style=black dot] (16) at (2, 0.75) {};
		\node [style=right label] (18) at (-0.25, 3) {$A$};
		\node [style=right label] (19) at (1.25, 3) {$B$};
		\node [style=right label] (20) at (2.25, 3) {$C$};
		\node [style=small box] (21) at (2, -0.25) {$\mat g$};
		\node [style=none] (22) at (2, -0.75) {};
		\node [style=none] (23) at (2, 0.25) {};
		\node [style=none] (24) at (0.75, 1.25) {};
		\node [style=none] (25) at (-0.5, -0.5) {};
		\node [style=none] (26) at (0.75, 3) {};
		\node [style=none] (27) at (-0.5, 3) {};
		\node [style=semilarge box] (28) at (-5, -0.25) {$\mat \omega$};
		\node [style=right label] (29) at (-5.75, 1) {$A$};
		\node [style=none] (30) at (-6, 0) {};
		\node [style=none] (31) at (-6, 1) {};
		\node [style=right label] (32) at (-4.75, 1) {$B$};
		\node [style=none] (33) at (-5, 0) {};
		\node [style=none] (34) at (-5, 1) {};
		\node [style=none] (35) at (-4, 1) {};
		\node [style=none] (36) at (-4, 0) {};
		\node [style=right label] (37) at (-3.75, 1) {$C$};
		\node [style=none] (38) at (-2.25, 0) {$=$};
	\end{pgfonlayer}
	\begin{pgfonlayer}{edgelayer}
		\draw (0.center) to (1.center);
		\draw (1.center) to (2.center);
		\draw (2.center) to (3.center);
		\draw (3.center) to (4.center);
		\draw (4.center) to (5.center);
		\draw (5.center) to (6.center);
		\draw (6.center) to (7.center);
		\draw (7.center) to (0.center);
		\draw (9) to (10.center);
		\draw (13.center) to (14.center);
		\draw (15.center) to (16);
		\draw (9) to (22.center);
		\draw (23.center) to (16);
		\draw [in=-90, out=180] (16) to (24.center);
		\draw [in=-90, out=180, looseness=0.75] (9) to (25.center);
		\draw (24.center) to (26.center);
		\draw (25.center) to (27.center);
		\draw (30.center) to (31.center);
		\draw (33.center) to (34.center);
		\draw (35.center) to (36.center);
	\end{pgfonlayer}
\end{tikzpicture}}
\end{equation}
\end{theorem}

\begin{proof}
The construction of $\mat f$ and $\mat g$ mimics the construction of
c-factors in~\cite{TianPearl}, using string diagrams and
(diagrammatic) disintegration.
Starting with a full-support $\mat \omega : I \to A \otimes B \otimes C$, we apply Theorem~\ref{thm:disintStoch} twice. First we can disintegrate $\mat \omega$ as $(\mat \omega' : I \to A \otimes B, \mat c : A \otimes B \to C)$ then further disintegrate $\mat \omega'$ into $(\mat a : I \to A, \mat b : A \to B)$:
\begin{equation}\label{eq:double-disint}
\begin{tikzpicture}
	\begin{pgfonlayer}{nodelayer}
		\node [style=none] (2) at (-1.75, 1) {};
		\node [style=none] (4) at (-3.75, 0) {$=$};
		\node [style=black dot] (7) at (-0.5, -0.25) {};
		\node [style=none] (9) at (-1.75, 2.5) {};
		\node [style=none] (10) at (1, 1) {};
		\node [style=none] (12) at (-6.5, 0.25) {};
		\node [style=none] (13) at (-6.5, 1.25) {};
		\node [style=none] (14) at (-5.5, 0.25) {};
		\node [style=none] (15) at (-5.5, 1.25) {};
		\node [style=right label] (16) at (-6.25, 1) {$B$};
		\node [style=right label] (17) at (-1.5, 2.5) {$A$};
		\node [style=right label] (19) at (-5.25, 1) {$C$};
		\node [style=none] (20) at (-7.5, 0.25) {};
		\node [style=none] (21) at (-7.5, 1.25) {};
		\node [style=right label] (22) at (-7.25, 1) {$A$};
		\node [style=none] (24) at (-0.75, 1) {};
		\node [style=black dot] (25) at (0.5, -0.25) {};
		\node [style=none] (26) at (-0.75, 2.5) {};
		\node [style=none] (27) at (1.5, 1) {};
		\node [style=right label] (28) at (-0.5, 2.5) {$B$};
		\node [style=small box] (30) at (1.25, 1.25) {$\mat c$};
		\node [style=none] (31) at (0.5, -1.25) {};
		\node [style=none] (32) at (-0.5, -1.25) {};
		\node [style=none] (33) at (1.25, 2.5) {};
		\node [style=right label] (34) at (1.5, 2.5) {$C$};
		\node [style=semilarge box] (40) at (-6.5, -0.25) {$\mat \omega$};
		\node [style=medium box] (42) at (0, -1.75) {$\mat \omega'$};
		\node [style=none] (43) at (4.5, 1.75) {};
		\node [style=none] (44) at (3.25, 0) {$=$};
		\node [style=black dot] (45) at (5.5, 0.5) {};
		\node [style=none] (46) at (4.5, 3.25) {};
		\node [style=none] (47) at (7.5, 1.75) {};
		\node [style=right label] (48) at (4.75, 3.25) {$A$};
		\node [style=none] (49) at (6, 1.75) {};
		\node [style=black dot] (50) at (7, 0.5) {};
		\node [style=none] (51) at (6, 3.25) {};
		\node [style=none] (52) at (8, 1.75) {};
		\node [style=right label] (53) at (6.25, 3.25) {$B$};
		\node [style=small box] (54) at (7.75, 2) {$\mat c$};
		\node [style=none] (57) at (7.75, 3.25) {};
		\node [style=right label] (58) at (8, 3.25) {$C$};
		\node [style=small box] (59) at (7, -1) {$\mat b$};
		\node [style=none] (60) at (5.5, -1.25) {};
		\node [style=black dot] (61) at (6.25, -2) {};
		\node [style=right label] (65) at (7.25, 0) {$B$};
		\node [style=none] (70) at (6.25, -3) {};
		\node [style=small box] (75) at (6.25, -3.5) {$\mat a$};
		\node [style=right label] (76) at (6.5, -2.5) {$A$};
		\node [style=none] (77) at (7, -1.25) {};
	\end{pgfonlayer}
	\begin{pgfonlayer}{edgelayer}
		\draw [in=-90, out=165] (7) to (2.center);
		\draw (2.center) to (9.center);
		\draw (12.center) to (13.center);
		\draw (14.center) to (15.center);
		\draw (20.center) to (21.center);
		\draw [in=-90, out=165] (25) to (24.center);
		\draw (24.center) to (26.center);
		\draw [in=-90, out=9, looseness=0.75] (7) to (10.center);
		\draw [in=-90, out=21] (25) to (27.center);
		\draw (7) to (32.center);
		\draw (25) to (31.center);
		\draw (33.center) to (30);
		\draw [in=-90, out=165] (45) to (43.center);
		\draw (43.center) to (46.center);
		\draw [in=-90, out=165] (50) to (49.center);
		\draw (49.center) to (51.center);
		\draw [in=-90, out=15, looseness=0.75] (45) to (47.center);
		\draw [in=-90, out=21] (50) to (52.center);
		\draw (57.center) to (54);
		\draw [in=-90, out=165] (61) to (60.center);
		\draw (61) to (70.center);
		\draw (60.center) to (45);
		\draw (59) to (50);
		\draw [in=-90, out=15] (61) to (77.center);
	\end{pgfonlayer}
\end{tikzpicture}
\end{equation}
Now, we let:
\[
\begin{tikzpicture}
	\begin{pgfonlayer}{nodelayer}
		\node [style=medium box] (0) at (0.5, 0) {$\mat f$};
		\node [style=none] (1) at (0, 0.25) {};
		\node [style=none] (2) at (1, 0.25) {};
		\node [style=none] (3) at (0.5, -0.5) {};
		\node [style=none] (4) at (0.5, -1.5) {};
		\node [style=none] (5) at (0, 1.5) {};
		\node [style=none] (6) at (1, 1.5) {};
		\node [style=right label] (7) at (0.25, 1.5) {$A$};
		\node [style=right label] (8) at (1.25, 1.5) {$C$};
		\node [style=right label] (9) at (0.75, -1.25) {$B$};
		\node [style=none] (10) at (2.25, 0) {$:=$};
		\node [style=none] (12) at (4, 2) {};
		\node [style=right label] (15) at (4.25, 2) {$A$};
		\node [style=small box] (21) at (5.75, 0.75) {$\mat c$};
		\node [style=none] (22) at (5.75, 2) {};
		\node [style=right label] (23) at (6, 2) {$C$};
		\node [style=none] (25) at (4, 0.25) {};
		\node [style=black dot] (26) at (4.75, -0.5) {};
		\node [style=none] (28) at (4.75, -1.5) {};
		\node [style=small box] (29) at (4.75, -2) {$\mat a$};
		\node [style=right label] (30) at (5, -1) {$A$};
		\node [style=none] (31) at (5.5, 0.25) {};
		\node [style=none] (32) at (6, 0.25) {};
		\node [style=right label] (33) at (6.25, -2.25) {$B$};
		\node [style=none] (34) at (6, -2.5) {};
	\end{pgfonlayer}
	\begin{pgfonlayer}{edgelayer}
		\draw (4.center) to (3.center);
		\draw (1.center) to (5.center);
		\draw (2.center) to (6.center);
		\draw (22.center) to (21);
		\draw [in=-90, out=165] (26) to (25.center);
		\draw (26) to (28.center);
		\draw (25.center) to (12.center);
		\draw [in=-90, out=15] (26) to (31.center);
		\draw (34.center) to (32.center);
	\end{pgfonlayer}
\end{tikzpicture}
\qquad\qquad
\begin{tikzpicture}
	\begin{pgfonlayer}{nodelayer}
		\node [style=small box] (0) at (0, 0) {$\mat g$};
		\node [style=none] (2) at (0, 0.25) {};
		\node [style=none] (3) at (0, -0.5) {};
		\node [style=none] (4) at (0, -1.5) {};
		\node [style=none] (6) at (0, 1.5) {};
		\node [style=right label] (8) at (0.25, 1.5) {$B$};
		\node [style=right label] (9) at (0.25, -1.25) {$A$};
		\node [style=none] (10) at (1.75, 0) {$:=$};
		\node [style=small box] (11) at (3.5, 0) {$\mat b$};
		\node [style=none] (12) at (3.5, 0.25) {};
		\node [style=none] (13) at (3.5, -0.5) {};
		\node [style=none] (14) at (3.5, -1.5) {};
		\node [style=none] (15) at (3.5, 1.5) {};
		\node [style=right label] (16) at (3.75, 1.5) {$B$};
		\node [style=right label] (17) at (3.75, -1.25) {$A$};
	\end{pgfonlayer}
	\begin{pgfonlayer}{edgelayer}
		\draw (4.center) to (3.center);
		\draw (2.center) to (6.center);
		\draw (14.center) to (13.center);
		\draw (12.center) to (15.center);
	\end{pgfonlayer}
\end{tikzpicture}
\]
Then \eqref{eq:comb-disint2} holds by construction of $\mat a, \mat b, \mat c$:
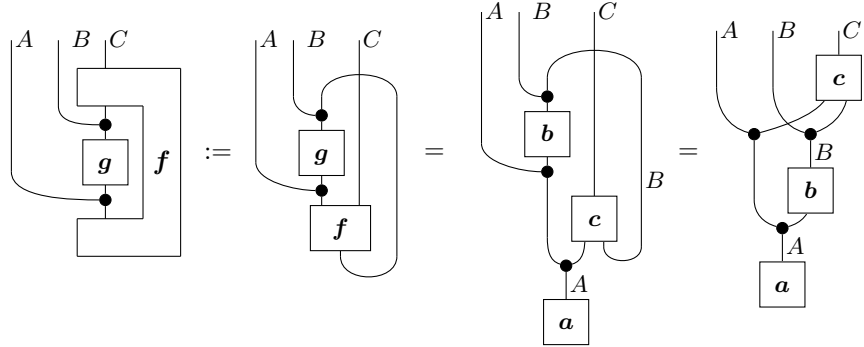
\[
\begin{tikzpicture}
	\begin{pgfonlayer}{nodelayer}
		\node [style=none] (0) at (-4.5, 2.25) {};
		\node [style=none] (1) at (-4.5, 1.25) {};
		\node [style=none] (2) at (-2.75, 1.25) {};
		\node [style=none] (3) at (-2.75, -1.75) {};
		\node [style=none] (4) at (-4.5, -1.75) {};
		\node [style=none] (5) at (-4.5, -2.75) {};
		\node [style=none] (6) at (-1.75, -2.75) {};
		\node [style=none] (7) at (-1.75, 2.25) {};
		\node [style=none] (8) at (-2.25, -0.25) {$\mat f$};
		\node [style=black dot] (9) at (-3.75, -1.25) {};
		\node [style=none] (10) at (-3.75, -1.75) {};
		\node [style=none] (13) at (-3.75, 3) {};
		\node [style=none] (14) at (-3.75, 2.25) {};
		\node [style=none] (15) at (-3.75, 1.25) {};
		\node [style=black dot] (16) at (-3.75, 0.75) {};
		\node [style=right label] (18) at (-6, 3) {$A$};
		\node [style=right label] (19) at (-4.5, 3) {$B$};
		\node [style=right label] (20) at (-3.5, 3) {$C$};
		\node [style=small box] (21) at (-3.75, -0.25) {$\mat g$};
		\node [style=none] (22) at (-3.75, -0.75) {};
		\node [style=none] (23) at (-3.75, 0.25) {};
		\node [style=none] (24) at (-5, 1.25) {};
		\node [style=none] (25) at (-6.25, -0.5) {};
		\node [style=none] (26) at (-5, 3) {};
		\node [style=none] (27) at (-6.25, 3) {};
		\node [style=none] (28) at (-0.75, 0) {$:=$};
		\node [style=medium box] (37) at (2.5, -2) {$\mat f$};
		\node [style=black dot] (38) at (2, -1) {};
		\node [style=none] (39) at (2, -1.5) {};
		\node [style=none] (40) at (3, 3) {};
		\node [style=none] (41) at (3, -1.5) {};
		\node [style=none] (42) at (2, 1.5) {};
		\node [style=black dot] (43) at (2, 1) {};
		\node [style=right label] (44) at (0.5, 3) {$A$};
		\node [style=right label] (45) at (1.75, 3) {$B$};
		\node [style=right label] (46) at (3.25, 3) {$C$};
		\node [style=small box] (47) at (2, 0) {$\mat g$};
		\node [style=none] (48) at (2, -0.5) {};
		\node [style=none] (49) at (2, 0.5) {};
		\node [style=none] (50) at (1.25, 1.5) {};
		\node [style=none] (51) at (0.25, -0.25) {};
		\node [style=none] (52) at (1.25, 3) {};
		\node [style=none] (53) at (0.25, 3) {};
		\node [style=none] (54) at (4, 1.5) {};
		\node [style=none] (55) at (4, -2.75) {};
		\node [style=none] (56) at (2.5, -2.75) {};
		\node [style=none] (57) at (2.5, -2.5) {};
		\node [style=small box] (60) at (9.25, -1.75) {$\mat c$};
		\node [style=right label] (62) at (9.5, 3.75) {$C$};
		\node [style=none] (63) at (8, -2.25) {};
		\node [style=black dot] (64) at (8.5, -3) {};
		\node [style=none] (65) at (8.5, -4) {};
		\node [style=small box] (66) at (8.5, -4.5) {$\mat a$};
		\node [style=right label] (67) at (8.75, -3.5) {$A$};
		\node [style=none] (68) at (9, -2.25) {};
		\node [style=right label] (70) at (10.75, -0.75) {$B$};
		\node [style=none] (72) at (5, 0) {$=$};
		\node [style=black dot] (74) at (8, -0.5) {};
		\node [style=none] (76) at (8, 2) {};
		\node [style=black dot] (77) at (8, 1.5) {};
		\node [style=right label] (78) at (6.5, 3.75) {$A$};
		\node [style=right label] (79) at (7.75, 3.75) {$B$};
		\node [style=small box] (80) at (8, 0.5) {$\mat b$};
		\node [style=none] (81) at (8, 0) {};
		\node [style=none] (82) at (8, 1) {};
		\node [style=none] (83) at (7.25, 2) {};
		\node [style=none] (84) at (6.25, 0.25) {};
		\node [style=none] (85) at (7.25, 3.75) {};
		\node [style=none] (86) at (6.25, 3.75) {};
		\node [style=none] (87) at (10.5, 2) {};
		\node [style=none] (88) at (10.5, -2.5) {};
		\node [style=none] (89) at (9.5, -2.5) {};
		\node [style=none] (90) at (9.5, -2.25) {};
		\node [style=none] (91) at (9.25, 3.75) {};
		\node [style=none] (92) at (9.25, -1.5) {};
		\node [style=none] (93) at (12.5, 1.75) {};
		\node [style=none] (94) at (11.75, 0) {$=$};
		\node [style=black dot] (95) at (13.5, 0.5) {};
		\node [style=none] (96) at (12.5, 3.25) {};
		\node [style=none] (97) at (15.5, 1.75) {};
		\node [style=right label] (98) at (12.75, 3.25) {$A$};
		\node [style=none] (99) at (14, 1.75) {};
		\node [style=black dot] (100) at (15, 0.5) {};
		\node [style=none] (101) at (14, 3.25) {};
		\node [style=none] (102) at (16, 1.75) {};
		\node [style=right label] (103) at (14.25, 3.25) {$B$};
		\node [style=small box] (104) at (15.75, 2) {$\mat c$};
		\node [style=none] (105) at (15.75, 3.25) {};
		\node [style=right label] (106) at (16, 3.25) {$C$};
		\node [style=small box] (107) at (15, -1) {$\mat b$};
		\node [style=none] (108) at (13.5, -1.25) {};
		\node [style=black dot] (109) at (14.25, -2) {};
		\node [style=right label] (110) at (15.25, 0) {$B$};
		\node [style=none] (111) at (14.25, -3) {};
		\node [style=small box] (112) at (14.25, -3.5) {$\mat a$};
		\node [style=right label] (113) at (14.5, -2.5) {$A$};
		\node [style=none] (114) at (15, -1.25) {};
	\end{pgfonlayer}
	\begin{pgfonlayer}{edgelayer}
		\draw (0.center) to (1.center);
		\draw (1.center) to (2.center);
		\draw (2.center) to (3.center);
		\draw (3.center) to (4.center);
		\draw (4.center) to (5.center);
		\draw (5.center) to (6.center);
		\draw (6.center) to (7.center);
		\draw (7.center) to (0.center);
		\draw (9) to (10.center);
		\draw (13.center) to (14.center);
		\draw (15.center) to (16);
		\draw (9) to (22.center);
		\draw (23.center) to (16);
		\draw [in=-90, out=180] (16) to (24.center);
		\draw [in=-90, out=180, looseness=0.75] (9) to (25.center);
		\draw (24.center) to (26.center);
		\draw (25.center) to (27.center);
		\draw (38) to (39.center);
		\draw (40.center) to (41.center);
		\draw (42.center) to (43);
		\draw (38) to (48.center);
		\draw (49.center) to (43);
		\draw [in=-90, out=180] (43) to (50.center);
		\draw [in=-90, out=180, looseness=0.75] (38) to (51.center);
		\draw (50.center) to (52.center);
		\draw (51.center) to (53.center);
		\draw [in=90, out=90] (42.center) to (54.center);
		\draw (54.center) to (55.center);
		\draw [in=-90, out=-90, looseness=1.25] (55.center) to (56.center);
		\draw (56.center) to (57.center);
		\draw [in=-90, out=165] (64) to (63.center);
		\draw (64) to (65.center);
		\draw [in=-90, out=15] (64) to (68.center);
		\draw (76.center) to (77);
		\draw (74) to (81.center);
		\draw (82.center) to (77);
		\draw [in=-90, out=180] (77) to (83.center);
		\draw [in=-90, out=180, looseness=0.75] (74) to (84.center);
		\draw (83.center) to (85.center);
		\draw (84.center) to (86.center);
		\draw [in=90, out=90] (76.center) to (87.center);
		\draw (87.center) to (88.center);
		\draw [in=-90, out=-90, looseness=1.25] (88.center) to (89.center);
		\draw (89.center) to (90.center);
		\draw (63.center) to (74);
		\draw (91.center) to (92.center);
		\draw [in=-90, out=165] (95) to (93.center);
		\draw (93.center) to (96.center);
		\draw [in=-90, out=165] (100) to (99.center);
		\draw (99.center) to (101.center);
		\draw [in=-90, out=15, looseness=0.75] (95) to (97.center);
		\draw [in=-90, out=21] (100) to (102.center);
		\draw (105.center) to (104);
		\draw [in=-90, out=165] (109) to (108.center);
		\draw (109) to (111.center);
		\draw (108.center) to (95);
		\draw (107) to (100);
		\draw [in=-90, out=15] (109) to (114.center);
	\end{pgfonlayer}
\end{tikzpicture}
\]
Note the last step above is just diagram deformation and the comonoid laws. The rightmost diagram above is equal to $\mat \omega$ by \eqref{eq:double-disint}.

For uniqueness, suppose \eqref{eq:comb-disint} holds for some other $\mat f', \mat g'$. Then by uniqueness of disintegration, it follows that $\mat g' = \mat g = \mat b$. To show that $\mat f = \mat f'$, we expand \eqref{eq:comb-disint2} explicitly in terms of matrices. This equation is equivalent to $\mat \omega^{ijk} = \mat f_j^{ik} \mat g_i^j = (\mat f')_j^{ik} \mat g_i^j$. Note that if $\mat g$ had any zero elements, $\mat \omega$ would not have full support, hence $\mat g_i^j \neq 0$ and therefore $\mat f_j^{ik} = (\mat f')_j^{ik}$ for all $i,j,k$. \QED
\end{proof}

Note that Theorem~\ref{thm:combuniqueextraction} generalises the normal
disintegration property given in Theorem~\ref{thm:disintStoch}. The
latter is recovered by taking $A := I$ (or $C := I$) above.


\section{Returning to the Smoking Scenario}\label{sec:smoking-return}
 
\begin{wrapfigure}{r}{0pt}
$\scalebox{0.8}{\begin{tikzpicture}
	\begin{pgfonlayer}{nodelayer}
		\node [style=large box] (0) at (0, 0) {$\mat \omega$};
		\node [style=none] (1) at (0, 2) {};
		\node [style=none] (2) at (1, 2) {};
		\node [style=none] (3) at (-1, 2) {};
		\node [style=none] (4) at (1, 0) {};
		\node [style=none] (5) at (-1, 0) {};
		\node [style=right label] (6) at (-0.75, 1.5) {$S$};
		\node [style=right label] (7) at (0.25, 1.5) {$T$};
		\node [style=right label] (8) at (1.25, 1.5) {$C$};
	\end{pgfonlayer}
	\begin{pgfonlayer}{edgelayer}
		\draw (1.center) to (0);
		\draw (2.center) to (4.center);
		\draw (3.center) to (5.center);
	\end{pgfonlayer}
\end{tikzpicture}}
\;=\; 
\scalebox{0.8}{\begin{tikzpicture}
	\begin{pgfonlayer}{nodelayer}
		\node [style=none] (1) at (0.75, 0.25) {};
		\node [style=none] (2) at (-2.25, 0.25) {};
		\node [style=black dot] (7) at (0.75, -0.75) {};
		\node [style=small box] (8) at (0.75, -2.25) {$\mat s$};
		\node [style=none] (9) at (-2.25, 5.75) {};
		\node [style=right label] (17) at (-2, 5.25) {$S$};
		\node [style=right label] (18) at (2, 5.25) {$C$};
		\node [style=semilarge box] (20) at (1.75, 3.75) {$\mat c$};
		\node [style=none] (21) at (2.75, -2.75) {};
		\node [style=small box] (22) at (1.75, -4.5) {$\mat h$};
		\node [style=black dot] (23) at (1.75, -3.5) {};
		\node [style=none] (24) at (2.75, 3.5) {};
		\node [style=none] (25) at (1.75, 4) {};
		\node [style=none] (26) at (1.75, 5.75) {};
		\node [style=right label] (27) at (3, -0.25) {$H$};
		\node [style=none] (28) at (0.75, -2.75) {};
		\node [style=small box] (29) at (0.75, 0.75) {$\mat t$};
		\node [style=none] (30) at (0.75, 3.5) {};
		\node [style=none] (31) at (-1, 3) {};
		\node [style=black dot] (32) at (0.75, 2.25) {};
		\node [style=none] (33) at (-1, 5.75) {};
		\node [style=right label] (34) at (-0.75, 5.25) {$T$};
		\node [style=none] (35) at (-0.25, -1.25) {};
		\node [style=none] (36) at (2.25, -1.25) {};
		\node [style=none] (37) at (2.25, 2.75) {};
		\node [style=none] (38) at (-0.25, 2.75) {};
		\node [style=none] (39) at (-0.25, 4.75) {};
		\node [style=none] (40) at (3.5, 4.75) {};
		\node [style=none] (41) at (3.5, -5.5) {};
		\node [style=none] (42) at (-0.25, -5.5) {};
		\node [style=none] (43) at (1.75, -0.25) {};
		\node [style=none] (44) at (-0.25, -0.25) {};
		\node [style=none] (45) at (-0.25, 1.75) {};
		\node [style=none] (46) at (1.75, 1.75) {};
		\node [style=none] (47) at (4, 2.75) {$\mat f$};
		\node [style=none] (48) at (-0.75, 0.75) {$\mat g$};
	\end{pgfonlayer}
	\begin{pgfonlayer}{edgelayer}
		\draw (7) to (8);
		\draw [in=-90, out=180, looseness=0.75] (7) to (2.center);
		\draw (2.center) to (9.center);
		\draw [in=-90, out=15] (23) to (21.center);
		\draw (21.center) to (24.center);
		\draw (25.center) to (26.center);
		\draw (22) to (23);
		\draw [in=-90, out=165] (23) to (28.center);
		\draw [in=-90, out=-180] (32) to (31.center);
		\draw (31.center) to (33.center);
		\draw (29) to (32);
		\draw (7) to (1.center);
		\draw (32) to (30.center);
		\draw [style=dashed edge] (41.center) to (42.center);
		\draw [style=dashed edge] (42.center) to (35.center);
		\draw [style=dashed edge] (35.center) to (36.center);
		\draw [style=dashed edge] (36.center) to (37.center);
		\draw [style=dashed edge] (37.center) to (38.center);
		\draw [style=dashed edge] (38.center) to (39.center);
		\draw [style=dashed edge] (39.center) to (40.center);
		\draw [style=dashed edge] (40.center) to (41.center);
		\draw [style=dashed edge] (43.center) to (44.center);
		\draw [style=dashed edge] (44.center) to (45.center);
		\draw [style=dashed edge] (45.center) to (46.center);
		\draw [style=dashed edge] (43.center) to (46.center);
	\end{pgfonlayer}
\end{tikzpicture}}$
\end{wrapfigure}
We now return to the smoking scenario of Section~\ref{sec:smoking}.
There, we concluded by claiming that the introduction of an
intermediate variable $T$ to the observational distribution
$\mat\omega : I \to S \otimes T \otimes C$
would enable us to calculate the interventional distribution. That
is, we can calculate
$\mat \omega' = \mathcal F(\syncut{S}(\omega))$ from
$\mat \omega := \mathcal F(\omega)$. Thanks to Theorem
\ref{thm:combuniqueextraction}, we are now able to perform that
calcuation. We first observe that our assumed causal structure for
$\mat \omega$ fits the form of Theorem~\ref{thm:combuniqueextraction},
where $\mat g$ is $\mat t$ and $\mat f$ is a 2-comb containing everything else,
as in the diagram on the side.

Hence, $\mat f$ and $\mat g$ are computable from $\mat \omega$. If we
plug them back together as in \eqref{eq:comb-disint},
we will get $\mat \omega$ back. However, if we insert a `cut' between $\mat f$
and $\mat g$:
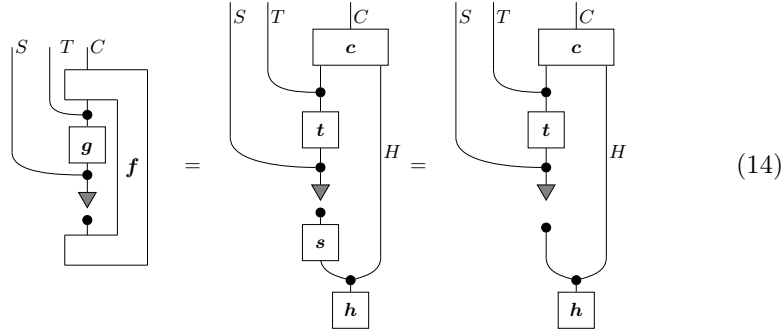
\begin{equation}\label{eq:cut-derivation}
\scalebox{0.8}{\begin{tikzpicture}
	\begin{pgfonlayer}{nodelayer}
		\node [style=none] (1) at (0, 0.75) {};
		\node [style=none] (2) at (-3, 1) {};
		\node [style=black dot] (7) at (0, 0) {};
		\node [style=small box] (8) at (0, -2.5) {$\mat s$};
		\node [style=none] (9) at (-3, 5.5) {};
		\node [style=right label] (17) at (-2.75, 5) {$S$};
		\node [style=right label] (18) at (1.25, 5) {$C$};
		\node [style=semilarge box] (20) at (1, 4) {$\mat c$};
		\node [style=none] (21) at (2, -3) {};
		\node [style=small box] (22) at (1, -4.75) {$\mat h$};
		\node [style=black dot] (23) at (1, -3.75) {};
		\node [style=none] (24) at (2, 3.75) {};
		\node [style=none] (25) at (1, 4.25) {};
		\node [style=none] (26) at (1, 5.5) {};
		\node [style=right label] (27) at (2.25, 0.5) {$H$};
		\node [style=none] (28) at (0, -3) {};
		\node [style=small box] (29) at (0, 1.25) {$\mat t$};
		\node [style=none] (30) at (0, 3.75) {};
		\node [style=none] (31) at (-1.75, 3.25) {};
		\node [style=black dot] (32) at (0, 2.5) {};
		\node [style=none] (33) at (-1.75, 5.5) {};
		\node [style=right label] (34) at (-1.5, 5) {$T$};
		\node [style=uniform] (35) at (0, -0.75) {};
		\node [style=black dot] (36) at (0, -1.5) {};
		\node [style=none] (37) at (3.25, 0) {$=$};
		\node [style=none] (38) at (7.5, 0.75) {};
		\node [style=none] (39) at (4.5, 1) {};
		\node [style=black dot] (40) at (7.5, 0) {};
		\node [style=none] (42) at (4.5, 5.5) {};
		\node [style=right label] (43) at (4.75, 5) {$S$};
		\node [style=right label] (44) at (8.75, 5) {$C$};
		\node [style=semilarge box] (45) at (8.5, 4) {$\mat c$};
		\node [style=none] (46) at (9.5, -3) {};
		\node [style=small box] (47) at (8.5, -4.75) {$\mat h$};
		\node [style=black dot] (48) at (8.5, -3.75) {};
		\node [style=none] (49) at (9.5, 3.75) {};
		\node [style=none] (50) at (8.5, 4.25) {};
		\node [style=none] (51) at (8.5, 5.5) {};
		\node [style=right label] (52) at (9.75, 0.5) {$H$};
		\node [style=none] (53) at (7.5, -3) {};
		\node [style=small box] (54) at (7.5, 1.25) {$\mat t$};
		\node [style=none] (55) at (7.5, 3.75) {};
		\node [style=none] (56) at (5.75, 3.25) {};
		\node [style=black dot] (57) at (7.5, 2.5) {};
		\node [style=none] (58) at (5.75, 5.5) {};
		\node [style=right label] (59) at (6, 5) {$T$};
		\node [style=uniform] (60) at (7.5, -0.75) {};
		\node [style=black dot] (61) at (7.5, -2) {};
		\node [style=none] (62) at (-8.5, 3.25) {};
		\node [style=none] (63) at (-8.5, 2.25) {};
		\node [style=none] (64) at (-6.75, 2.25) {};
		\node [style=none] (65) at (-6.75, -2.25) {};
		\node [style=none] (66) at (-8.5, -2.25) {};
		\node [style=none] (67) at (-8.5, -3.25) {};
		\node [style=none] (68) at (-5.75, -3.25) {};
		\node [style=none] (69) at (-5.75, 3.25) {};
		\node [style=none] (70) at (-6.25, 0) {$\mat f$};
		\node [style=black dot] (71) at (-7.75, -0.25) {};
		\node [style=none] (72) at (-7.75, -2.25) {};
		\node [style=none] (73) at (-7.75, 4) {};
		\node [style=none] (74) at (-7.75, 3.25) {};
		\node [style=none] (75) at (-7.75, 2.25) {};
		\node [style=black dot] (76) at (-7.75, 1.75) {};
		\node [style=right label] (77) at (-10, 4) {$S$};
		\node [style=right label] (78) at (-8.5, 4) {$T$};
		\node [style=right label] (79) at (-7.5, 4) {$C$};
		\node [style=small box] (80) at (-7.75, 0.75) {$\mat g$};
		\node [style=none] (81) at (-7.75, 0.25) {};
		\node [style=none] (82) at (-7.75, 1.25) {};
		\node [style=none] (83) at (-9, 2.25) {};
		\node [style=none] (84) at (-10.25, 0.5) {};
		\node [style=none] (85) at (-9, 4) {};
		\node [style=none] (86) at (-10.25, 4) {};
		\node [style=none] (87) at (-4.25, 0) {$=$};
		\node [style=uniform] (88) at (-7.75, -1) {};
		\node [style=black dot] (89) at (-7.75, -1.75) {};
	\end{pgfonlayer}
	\begin{pgfonlayer}{edgelayer}
		\draw [in=-90, out=180, looseness=0.75] (7) to (2.center);
		\draw (2.center) to (9.center);
		\draw [in=-90, out=15] (23) to (21.center);
		\draw (21.center) to (24.center);
		\draw (25.center) to (26.center);
		\draw (22) to (23);
		\draw [in=-90, out=165] (23) to (28.center);
		\draw [in=-90, out=-180] (32) to (31.center);
		\draw (31.center) to (33.center);
		\draw (29) to (32);
		\draw (7) to (1.center);
		\draw (32) to (30.center);
		\draw (35) to (7);
		\draw (8) to (36);
		\draw [in=-90, out=180, looseness=0.75] (40) to (39.center);
		\draw (39.center) to (42.center);
		\draw [in=-90, out=15] (48) to (46.center);
		\draw (46.center) to (49.center);
		\draw (50.center) to (51.center);
		\draw (47) to (48);
		\draw [in=-90, out=165] (48) to (53.center);
		\draw [in=-90, out=-180] (57) to (56.center);
		\draw (56.center) to (58.center);
		\draw (54) to (57);
		\draw (40) to (38.center);
		\draw (57) to (55.center);
		\draw (60) to (40);
		\draw (53.center) to (61);
		\draw (62.center) to (63.center);
		\draw (63.center) to (64.center);
		\draw (64.center) to (65.center);
		\draw (65.center) to (66.center);
		\draw (66.center) to (67.center);
		\draw (67.center) to (68.center);
		\draw (68.center) to (69.center);
		\draw (69.center) to (62.center);
		\draw (73.center) to (74.center);
		\draw (75.center) to (76);
		\draw (71) to (81.center);
		\draw (82.center) to (76);
		\draw [in=-90, out=180] (76) to (83.center);
		\draw [in=-90, out=180, looseness=0.75] (71) to (84.center);
		\draw (83.center) to (85.center);
		\draw (84.center) to (86.center);
		\draw (72.center) to (89);
		\draw (88) to (71);
	\end{pgfonlayer}
\end{tikzpicture}}
\end{equation}
we obtain $\mat \omega' = \mathcal F(\syncut{S}(\omega))$.

We now consider a concrete example. Fix interpretations $S = T = C = \{0,1\}$ and let
$\mat \omega : I \to S \otimes T \otimes C$ be the stochastic matrix:
\[
\mat \omega :=
\left(
\begin{matrix}
   0.5  \\[-0.5em]  
   0.1  \\[-0.5em]  
   0.01 \\[-0.5em]  
   0.02 \\[-0.5em]  
   0.1  \\[-0.5em]  
   0.05 \\[-0.5em]  
   0.02 \\[-0.5em]  
   0.2      
\end{matrix}
\right)\ 
{\color{gray}
\begin{smallmatrix}
  \textrm{\footnotesize $\leftarrow P(S=0,T=0,C=0)$} \\
  \textrm{\footnotesize $\leftarrow P(S=0,T=0,C=1)$} \\
  \textrm{\footnotesize $\leftarrow P(S=0,T=1,C=0)$} \\
  \textrm{\footnotesize $\leftarrow P(S=0,T=1,C=1)$} \\
  \textrm{\footnotesize $\leftarrow P(S=1,T=0,C=0)$} \\
  \textrm{\footnotesize $\leftarrow P(S=1,T=0,C=1)$} \\
  \textrm{\footnotesize $\leftarrow P(S=1,T=1,C=0)$} \\
  \textrm{\footnotesize $\leftarrow P(S=1,T=1,C=1)$}
\end{smallmatrix}}
\]

\noindent Now, disintegrating $\mat \omega$:
\[
\scalebox{0.8}{\begin{tikzpicture}
	\begin{pgfonlayer}{nodelayer}
		\node [style=none] (4) at (0, 0) {$=$};
		\node [style=none] (12) at (-2.75, 0.25) {};
		\node [style=black dot] (13) at (-2.75, 1.5) {};
		\node [style=none] (14) at (-1.75, 0.25) {};
		\node [style=none] (15) at (-1.75, 1.25) {};
		\node [style=right label] (16) at (-2.5, 1) {$T$};
		\node [style=right label] (19) at (-1.5, 1) {$C$};
		\node [style=none] (20) at (-3.75, 0.25) {};
		\node [style=none] (21) at (-3.75, 1.25) {};
		\node [style=right label] (22) at (-3.5, 1) {$S$};
		\node [style=none] (23) at (-1.75, 1.25) {};
		\node [style=semilarge box] (41) at (-2.75, -0.25) {$\mat \omega$};
		\node [style=small box] (42) at (3.5, 0.75) {$\mat c$};
		\node [style=none] (43) at (1.5, 0.75) {};
		\node [style=black dot] (44) at (2.5, -0.5) {};
		\node [style=small box] (45) at (2.5, -1.75) {$\mat s$};
		\node [style=none] (46) at (1.5, 2.25) {};
		\node [style=none] (47) at (3.5, 2.25) {};
		\node [style=right label] (48) at (1.75, 1.75) {$S$};
		\node [style=right label] (49) at (3.75, 1.75) {$C$};
	\end{pgfonlayer}
	\begin{pgfonlayer}{edgelayer}
		\draw (12.center) to (13);
		\draw (14.center) to (15.center);
		\draw (20.center) to (21.center);
		\draw (44) to (45);
		\draw [in=-90, out=15] (44) to (42);
		\draw [in=-90, out=165] (44) to (43.center);
		\draw (43.center) to (46.center);
		\draw (42) to (47.center);
	\end{pgfonlayer}
\end{tikzpicture}}
\qquad\textrm{gives}\qquad
\mat c \approx
\left(\begin{matrix}0.81 & 0.32\\0.19 & 0.68\end{matrix}\right)
\]
The bottom-left element of $\mat c$ is $P(C=1|S=0)$, whereas the bottom-right is $P(C=1|S=1)$, so this suggests that patients are $\approx 3.5$ times as likely to get cancer if they smoke ($68 \%$ vs. $19\%$). However, comb-disintegrating $\mat \omega$ using Theorem~\ref{thm:combuniqueextraction} gives $\mat g \colon S \to T$ and a comb $\mat f \colon T \to S \otimes C$ with the following stochastic matrices:
\[
\mat f \approx \left(\begin{matrix}0.53 & 0.21\\0.11 & 0.42\\0.25 & 0.03\\0.12 & 0.34 \end{matrix}\right)
\qquad
\qquad
\mat g \approx \left(\begin{matrix}0.95 & 0.41\\0.05 & 0.59\end{matrix}\right)
\]
Recomposing these with a `cut' in between, as in the left-hand side of \eqref{eq:cut-derivation}, gives the interventional distribution $\mat \omega' \approx (0.38, 0.11, 0.01, 0.02, 0.16, 0.05, 0.07, 0.22)$. Disintegrating:
\[
\scalebox{0.8}{\begin{tikzpicture}
	\begin{pgfonlayer}{nodelayer}
		\node [style=none] (4) at (0, 0) {$=$};
		\node [style=none] (12) at (-2.75, 0.25) {};
		\node [style=black dot] (13) at (-2.75, 1.5) {};
		\node [style=none] (14) at (-1.75, 0.25) {};
		\node [style=none] (15) at (-1.75, 1.25) {};
		\node [style=right label] (16) at (-2.5, 1) {$T$};
		\node [style=right label] (19) at (-1.5, 1) {$C$};
		\node [style=none] (20) at (-3.75, 0.25) {};
		\node [style=none] (21) at (-3.75, 1.25) {};
		\node [style=right label] (22) at (-3.5, 1) {$S$};
		\node [style=none] (23) at (-1.75, 1.25) {};
		\node [style=semilarge box] (41) at (-2.75, -0.25) {$\mat \omega'$};
		\node [style=small box] (42) at (3.5, 0.75) {$\mat c'$};
		\node [style=none] (43) at (1.5, 0.75) {};
		\node [style=black dot] (44) at (2.5, -0.5) {};
		\node [style=small box] (45) at (2.5, -1.75) {$\mat s'$};
		\node [style=none] (46) at (1.5, 2.25) {};
		\node [style=none] (47) at (3.5, 2.25) {};
		\node [style=right label] (48) at (1.75, 1.75) {$S$};
		\node [style=right label] (49) at (3.75, 1.75) {$C$};
	\end{pgfonlayer}
	\begin{pgfonlayer}{edgelayer}
		\draw (12.center) to (13);
		\draw (14.center) to (15.center);
		\draw (20.center) to (21.center);
		\draw (44) to (45);
		\draw [in=-90, out=15] (44) to (42);
		\draw [in=-90, out=165] (44) to (43.center);
		\draw (43.center) to (46.center);
		\draw (42) to (47.center);
	\end{pgfonlayer}
\end{tikzpicture}}
\qquad\textrm{gives}\qquad
\mat c' \approx
\left(\begin{matrix}0.75 & 0.46\\0.25 & 0.54\end{matrix}\right).
\]
From the interventional distribution, we conclude that, in a
(hypothotetical) clinical trial, patients are about twice as likely to
get cancer if they smoke ($54 \%$ vs. $25\%$). So, since $54 < 68$, there
was \textit{some} confounding influence between $S$ and $C$ in our
observational data, but after removing it via comb disintegration, we
see there is still a signficant causal link between smoking and
cancer.

Note this conclusion depends totally on the particular observational
data that we picked. For a different interpretation of $\mat \omega$ in
$\Stoch$, one might conclude that there is \emph{no} causal connection, or even that smoking \textit{decreases} the chance of
getting cancer. Interestingly, all three cases can arise even when a na\"ive analysis of the data shows a strong direct correlation between $S$ and $C$. To see and/or experiment with these cases, we have provided the Python code\footnote{\url{https://gist.github.com/akissinger/aeec1751792a208253bda491ead587b6}} used to perform these calculations. See also \cite{NielsenBlog} for a pedagocical overview of this example (using traditional Bayesian network language) with some sample calculations.
 
\section{The General Case for a Single Intervention}\label{sec:general}

While we applied the comb decomposition to a particular example, this
technique applies essentially unmodified to many examples where we
intervene at a single variable (called $X$ below) within an arbitrary
causal structure.

\begin{theorem}\label{thm:single-intervention}
  Let $G$ be a dag with a fixed node $X$ that has corresponding generator $x \colon Y_1 \otimes \ldots \otimes Y_n \to X$ in $\free{G}$. Then, suppose $\omega$ is a morphism in $\free{G}$ of the following form:
  \begin{equation}\label{eq:syn-x}
    \scalebox{0.8}{\begin{tikzpicture}
	\begin{pgfonlayer}{nodelayer}
		\node [style=large box, minimum width=2cm] (28) at (-5.5, -0.25) {$\omega$};
		\node [style=right label] (29) at (-6.75, 1) {$A$};
		\node [style=none] (30) at (-7, 0) {};
		\node [style=none] (31) at (-7, 1) {};
		\node [style=right label] (32) at (-4.75, 1) {$B$};
		\node [style=none] (33) at (-5, 0) {};
		\node [style=none] (34) at (-5, 1) {};
		\node [style=none] (35) at (-4, 1) {};
		\node [style=none] (36) at (-4, 0) {};
		\node [style=right label] (37) at (-3.75, 1) {$C$};
		\node [style=none] (38) at (-2.25, 0) {$=$};
		\node [style=right label] (39) at (-5.75, 1) {$X$};
		\node [style=none] (40) at (-6, 0) {};
		\node [style=none] (41) at (-6, 1) {};
		\node [style=black dot] (54) at (2.75, -1.5) {};
		\node [style=none] (55) at (2.75, -3.75) {};
		\node [style=none] (56) at (4, 3) {};
		\node [style=none] (57) at (4, 2.25) {};
		\node [style=right label] (60) at (0.75, 3) {$X$};
		\node [style=right label] (62) at (4.25, 3) {$C$};
		\node [style=medium box] (63) at (3.25, 0) {$g$};
		\node [style=none] (64) at (2.75, -0.5) {};
		\node [style=none] (67) at (-0.5, 0.75) {};
		\node [style=none] (69) at (-0.5, 3) {};
		\node [style=black dot] (72) at (3.75, -1) {};
		\node [style=none] (73) at (3.75, -3.75) {};
		\node [style=right label] (74) at (-0.25, 3) {$A$};
		\node [style=none] (75) at (3.75, -0.5) {};
		\node [style=none] (76) at (0.5, 1) {};
		\node [style=none] (77) at (0.5, 3) {};
		\node [style=right label] (78) at (2, 3) {$B$};
		\node [style=black dot] (79) at (3.25, 1) {};
		\node [style=none] (80) at (3.25, 0.5) {};
		\node [style=none] (81) at (3.25, 1.5) {};
		\node [style=none] (82) at (1.75, 1.75) {};
		\node [style=none] (83) at (1.75, 3) {};
		\node [style=small box] (84) at (3.75, -2.5) {$x$};
		\node [style=semilarge box] (85) at (3.75, -4) {$f_1$};
		\node [style=semilarge box] (86) at (4, 2) {$f_2$};
		\node [style=none] (87) at (4.75, -3.75) {};
		\node [style=none] (88) at (4.75, 1.5) {};
	\end{pgfonlayer}
	\begin{pgfonlayer}{edgelayer}
		\draw (30.center) to (31.center);
		\draw (33.center) to (34.center);
		\draw (35.center) to (36.center);
		\draw (40.center) to (41.center);
		\draw (56.center) to (57.center);
		\draw (54) to (64.center);
		\draw [in=-90, out=180] (54) to (67.center);
		\draw (67.center) to (69.center);
		\draw (72) to (75.center);
		\draw [in=-90, out=180] (72) to (76.center);
		\draw (76.center) to (77.center);
		\draw (55.center) to (54);
		\draw (79) to (81.center);
		\draw [in=-90, out=180] (79) to (82.center);
		\draw (82.center) to (83.center);
		\draw (80.center) to (79);
		\draw (73.center) to (84);
		\draw [in=270, out=90] (84) to (72);
		\draw (87.center) to (88.center);
	\end{pgfonlayer}
\end{tikzpicture}}
  \end{equation}
  for some morphisms $f_1, f_2$ and $g$ in $\free{G}$ not containing $x$ as a subdiagram.
  Then the interventional distribution $\mat \omega' := \mathcal F(\syncut{X}(\omega))$ is computable from the observational distribution $\mat \omega = \mathcal F(\omega)$.
\end{theorem}

\begin{proof}
  The proof is very close to the example in the previous section. Interpreting $\omega$ into \Stoch, we get a diagram of stochastic maps, which we can comb-disintegrate, then recompose with $\uniform \circ \counit$ to produce the interventional distribution:
  \begin{equation*}
    \scalebox{0.8}
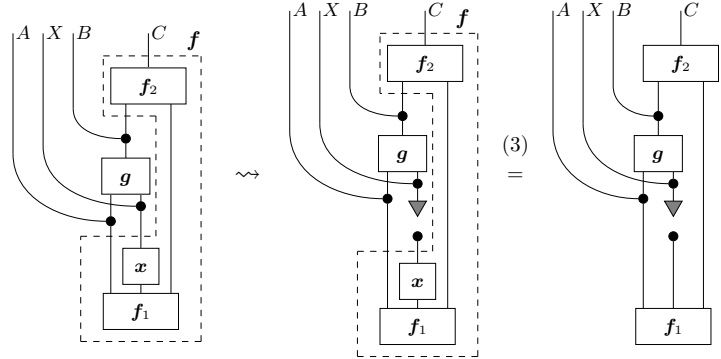
{\begin{tikzpicture}
	\begin{pgfonlayer}{nodelayer}
		\node [style=black dot] (54) at (0, -1.5) {};
		\node [style=none] (55) at (0, -4.25) {};
		\node [style=none] (56) at (1.25, 4.75) {};
		\node [style=none] (57) at (1.25, 3.5) {};
		\node [style=right label] (60) at (-2, 4.75) {$X$};
		\node [style=right label] (62) at (1.5, 4.75) {$C$};
		\node [style=medium box] (63) at (0.5, 0) {$\mat g$};
		\node [style=none] (64) at (0, -0.5) {};
		\node [style=none] (67) at (-3.25, 0.75) {};
		\node [style=none] (69) at (-3.25, 4.75) {};
		\node [style=black dot] (72) at (1, -1) {};
		\node [style=none] (73) at (1, -4.25) {};
		\node [style=right label] (74) at (-3, 4.75) {$A$};
		\node [style=none] (75) at (1, -0.5) {};
		\node [style=none] (76) at (-2.25, 1) {};
		\node [style=none] (77) at (-2.25, 4.75) {};
		\node [style=right label] (78) at (-1, 4.75) {$B$};
		\node [style=black dot] (79) at (0.5, 1.25) {};
		\node [style=none] (80) at (0.5, 0.5) {};
		\node [style=none] (81) at (0.5, 2.5) {};
		\node [style=none] (82) at (-1.25, 2.25) {};
		\node [style=none] (83) at (-1.25, 4.75) {};
		\node [style=small box] (84) at (1, -3) {$\mat x$};
		\node [style=semilarge box] (85) at (1, -4.5) {$\mat f_1$};
		\node [style=semilarge box] (86) at (1.25, 3) {$\mat f_2$};
		\node [style=none] (87) at (2, -4.25) {};
		\node [style=none] (88) at (2, 2.5) {};
		\node [style=none] (89) at (-1, -2) {};
		\node [style=none] (90) at (1.5, -2) {};
		\node [style=none] (91) at (1.5, 2) {};
		\node [style=none] (92) at (-0.25, 2) {};
		\node [style=none] (93) at (-0.25, 4) {};
		\node [style=none] (94) at (3, -5.5) {};
		\node [style=none] (95) at (3, 4) {};
		\node [style=none] (96) at (-1, -5.5) {};
		\node [style=none] (97) at (2.75, 4.5) {$\mat f$};
	\end{pgfonlayer}
	\begin{pgfonlayer}{edgelayer}
		\draw (56.center) to (57.center);
		\draw (54) to (64.center);
		\draw [in=-90, out=180] (54) to (67.center);
		\draw (67.center) to (69.center);
		\draw (72) to (75.center);
		\draw [in=-90, out=180] (72) to (76.center);
		\draw (76.center) to (77.center);
		\draw (55.center) to (54);
		\draw (79) to (81.center);
		\draw [in=-90, out=180] (79) to (82.center);
		\draw (82.center) to (83.center);
		\draw (80.center) to (79);
		\draw (73.center) to (84);
		\draw (84) to (72);
		\draw (87.center) to (88.center);
		\draw [style=dashed edge] (89.center) to (90.center);
		\draw [style=dashed edge] (90.center) to (91.center);
		\draw [style=dashed edge] (91.center) to (92.center);
		\draw [style=dashed edge] (92.center) to (93.center);
		\draw [style=dashed edge] (94.center) to (95.center);
		\draw [style=dashed edge] (94.center) to (96.center);
		\draw [style=dashed edge] (96.center) to (89.center);
		\draw [style=dashed edge] (93.center) to (95.center);
	\end{pgfonlayer}
\end{tikzpicture}} \ \ \leadsto\ \ 
    \scalebox{0.8}{\begin{tikzpicture}
	\begin{pgfonlayer}{nodelayer}
		\node [style=black dot] (54) at (0, -0.75) {};
		\node [style=none] (55) at (0, -4.75) {};
		\node [style=none] (56) at (1.25, 5.5) {};
		\node [style=none] (57) at (1.25, 4.25) {};
		\node [style=right label] (60) at (-2, 5.5) {$X$};
		\node [style=right label] (62) at (1.5, 5.5) {$C$};
		\node [style=medium box] (63) at (0.5, 0.75) {$\mat g$};
		\node [style=none] (64) at (0, 0.25) {};
		\node [style=none] (67) at (-3.25, 1.5) {};
		\node [style=none] (69) at (-3.25, 5.5) {};
		\node [style=black dot] (72) at (1, -0.25) {};
		\node [style=none] (73) at (1, -4.75) {};
		\node [style=right label] (74) at (-3, 5.5) {$A$};
		\node [style=none] (75) at (1, 0.25) {};
		\node [style=none] (76) at (-2.25, 1.75) {};
		\node [style=none] (77) at (-2.25, 5.5) {};
		\node [style=right label] (78) at (-1, 5.5) {$B$};
		\node [style=black dot] (79) at (0.5, 2) {};
		\node [style=none] (80) at (0.5, 1.25) {};
		\node [style=none] (81) at (0.5, 3.25) {};
		\node [style=none] (82) at (-1.25, 3) {};
		\node [style=none] (83) at (-1.25, 5.5) {};
		\node [style=small box] (84) at (1, -3.5) {$\mat x$};
		\node [style=semilarge box] (85) at (1, -5) {$\mat f_1$};
		\node [style=semilarge box] (86) at (1.25, 3.75) {$\mat f_2$};
		\node [style=none] (87) at (2, -4.75) {};
		\node [style=none] (88) at (2, 3.25) {};
		\node [style=none] (89) at (-1, -2.5) {};
		\node [style=none] (90) at (1.5, -2.5) {};
		\node [style=none] (91) at (1.5, 2.75) {};
		\node [style=none] (92) at (-0.25, 2.75) {};
		\node [style=none] (93) at (-0.25, 4.75) {};
		\node [style=none] (94) at (3, -6) {};
		\node [style=none] (95) at (3, 4.75) {};
		\node [style=none] (96) at (-1, -6) {};
		\node [style=none] (97) at (2.5, 5.25) {$\mat f$};
		\node [style=uniform] (98) at (1, -1) {};
		\node [style=black dot] (99) at (1, -2) {};
		\node [style=none] (100) at (4.25, 0) {$=$};
		\node [style=black dot] (101) at (8.5, -0.75) {};
		\node [style=none] (102) at (8.5, -4.75) {};
		\node [style=none] (103) at (9.75, 5.5) {};
		\node [style=none] (104) at (9.75, 4.25) {};
		\node [style=right label] (105) at (6.75, 5.5) {$X$};
		\node [style=right label] (106) at (10, 5.5) {$C$};
		\node [style=medium box] (107) at (9, 0.75) {$\mat g$};
		\node [style=none] (108) at (8.5, 0.25) {};
		\node [style=none] (109) at (5.5, 1.5) {};
		\node [style=none] (110) at (5.5, 5.5) {};
		\node [style=black dot] (111) at (9.5, -0.25) {};
		\node [style=none] (112) at (9.5, -4.75) {};
		\node [style=right label] (113) at (5.75, 5.5) {$A$};
		\node [style=none] (114) at (9.5, 0.25) {};
		\node [style=none] (115) at (6.5, 1.75) {};
		\node [style=none] (116) at (6.5, 5.5) {};
		\node [style=right label] (117) at (7.75, 5.5) {$B$};
		\node [style=black dot] (118) at (9, 2) {};
		\node [style=none] (119) at (9, 1.25) {};
		\node [style=none] (120) at (9, 3.25) {};
		\node [style=none] (121) at (7.5, 3) {};
		\node [style=none] (122) at (7.5, 5.5) {};
		\node [style=semilarge box] (124) at (9.5, -5) {$\mat f_1$};
		\node [style=semilarge box] (125) at (9.75, 3.75) {$\mat f_2$};
		\node [style=none] (126) at (10.5, -4.75) {};
		\node [style=none] (127) at (10.5, 3.25) {};
		\node [style=uniform] (137) at (9.5, -1) {};
		\node [style=black dot] (138) at (9.5, -2) {};
		\node [style=none] (139) at (4.25, 1) {\eqref{eq:discarding-final}};
	\end{pgfonlayer}
	\begin{pgfonlayer}{edgelayer}
		\draw (56.center) to (57.center);
		\draw (54) to (64.center);
		\draw [in=-90, out=180] (54) to (67.center);
		\draw (67.center) to (69.center);
		\draw (72) to (75.center);
		\draw [in=-90, out=180] (72) to (76.center);
		\draw (76.center) to (77.center);
		\draw (55.center) to (54);
		\draw (79) to (81.center);
		\draw [in=-90, out=180] (79) to (82.center);
		\draw (82.center) to (83.center);
		\draw (80.center) to (79);
		\draw (73.center) to (84);
		\draw (87.center) to (88.center);
		\draw [style=dashed edge] (89.center) to (90.center);
		\draw [style=dashed edge] (90.center) to (91.center);
		\draw [style=dashed edge] (91.center) to (92.center);
		\draw [style=dashed edge] (92.center) to (93.center);
		\draw [style=dashed edge] (94.center) to (95.center);
		\draw [style=dashed edge] (94.center) to (96.center);
		\draw [style=dashed edge] (96.center) to (89.center);
		\draw [style=dashed edge] (93.center) to (95.center);
		\draw (84) to (99);
		\draw (98) to (72);
		\draw (103.center) to (104.center);
		\draw (101) to (108.center);
		\draw [in=-90, out=180] (101) to (109.center);
		\draw (109.center) to (110.center);
		\draw (111) to (114.center);
		\draw [in=-90, out=180] (111) to (115.center);
		\draw (115.center) to (116.center);
		\draw (102.center) to (101);
		\draw (118) to (120.center);
		\draw [in=-90, out=180] (118) to (121.center);
		\draw (121.center) to (122.center);
		\draw (119.center) to (118);
		\draw (126.center) to (127.center);
		\draw (137) to (111);
		\draw (112.center) to (138);
	\end{pgfonlayer}
\end{tikzpicture}}
  \end{equation*}
  The RHS above is then $\mathcal F(\syncut{X}(\omega))$. \QED
\end{proof}

This is general enough to cover several well-known sufficient
conditions from the causality literature, including single-variable
versions of the so-called \textit{front-door} and \textit{back-door}
criteria, as well as the sufficient condition based on confounding
paths given by Pearl and Tian~\cite{TianPearl}. As the latter subsumes
the other two, we will say a few words about the relationship between
the Pearl/Tian condition and Theorem~\ref{thm:single-intervention}. In
\cite{TianPearl}, the authors focus on \textit{semi-Markovian} models,
where the only latent variables have exactly two observed children and
no parents. Suppose we write $A \leftrightarrow B$ if two observed
variables are connected by a latent common cause, then one can
characterise \textit{confounding paths} as the transitive closure of
$\leftrightarrow$. They go on to show that the interventional
distribution corresponding cutting $X$ is computable whenever there
are no confounding paths connecting $X$ to one of its children.

We can compare this to the form of expression $\omega$ in
equation~\eqref{eq:syn-x}. First, note this factorisation implies that
all boxes which take $X$ as an input must occur as sub-diagrams of
$g$. Hence, any `confounding path' connecting $X$ to its children
would yield at least one (un-copied) wire from $f_1$ to $g$, hence it
cannot be factored as \eqref{eq:syn-x}. Conversely, if there are no
confounding paths from $X$ to its children, then we can we can place
the boxes involved in any other confounding path either entirely
inside of $g$ or entirely outside of $g$ and obtain factorisation
\eqref{eq:syn-x}. Hence, restricting to semi-Markovian models, the no-
confounding-path condition from \cite{TianPearl} is equivalent to
ours. However, Theorem~\ref{thm:single-intervention} is slightly more
general: its formulation doesn't rely on the causal structure $\omega$
being semi-Markovian.

\vspace{-2mm}
\section{Conclusion and future work}\label{sec:conclusion}
\vspace{-1mm}

This paper takes a fresh, systematic look at the problem of causal
identifiability. By clearly distinguishing syntax (string
diagram surgery and identification of comb shapes) and semantics (comb-
disintegration of joint states) we obtain a clear methodology for computing
interventional distributions, and hence causal effects, from observational
data.

A natural next step is moving beyond single-variable interventions to
the general case, i.e. situations where we allow interventions on
multiple variables which may have some arbitrary causal relationships
connecting them. This would mean extending the comb factorisation
Theorem~\ref{thm:combuniqueextraction} from a 2-comb and a channel to
arbitrary $n$-combs. This seems to be straightforward, via an inductive
extension of the proof of Theorem~\ref{thm:combuniqueextraction}.
A more substantial direction of future work will be the
strengthening of Theorem~\ref{thm:single-intervention} from sufficient
conditions for causal identifiability to a full characterisation.
Indeed, the related condition based on confounding paths from
\cite{TianPearl} is a necessary and sufficient condition for computing
the interventional distribution on a single variable. Hence, it will
be interesting to formalise this necessity proof (and more general
versions, e.g.~\cite{Huang2008}) within our framework and investigate,
for example, the extent to which it holds beyond the semi-Markovian
case.

While we focus exclusively on the case of taking models in \Stoch in this
paper, the techniques we gave are posed at an abstract level in terms of
composition and factorisation. Hence, we are optimistic about their prospects
to generalise to other probabilistic (e.g. infinite discrete and continuous
variables) and quantum settings. In the latter case, this could provide
insights into the emerging field of
\textit{quantum causal structures}~\cite{CostaShrapnel,PienaarBrukner,LeiferSpekkens,SpekkensInfer,HensonLalPusey}, which attempts in part to replay some of the results coming from statistical causal reasoning, but where quantum processes play a role
analogous to stochastic ones. A key difficulty in applying our framework to a category of quantum processes, rather than \Stoch, is the unavailability of `copy' morphisms due to the
quantum no-cloning theorem~\cite{NoCloning}. However, a recent proposal for
the formulation of `quantum common causes'~\cite{AllenCommonCause} suggests a
(partially-defined) analogue to the role played by `copy' in our
formulation constructed via multiplication of certain commuting Choi
matrices. Hence, it may yet be possible to import results from classical
causal reasoning into the quantum case just by changing the category of models.

\medskip

{\noindent\textbf{Acknowledgements.} FZ acknowledges support from EPSRC grant EP/R020604/1. AK would like to thank Tom Claassen and Elie Wolfe for useful discussions on causal identification criteria.}

\newpage

\bibliographystyle{plain}
\bibliography{main}



\end{document}